\documentclass[12pt]{amsart}
\usepackage{amsmath,amsfonts,amssymb}
\usepackage{amsthm,amscd}
\newcommand{\A}{\ensuremath{\mathcal{A}}}
\newcommand{\B}{\ensuremath{\mathcal{B}}}

\newcommand{\R}{\ensuremath{\mathcal{R}}}
\newcommand{\KC}{\ensuremath{{\mathbb C}}}
\newcommand{\KR}{\ensuremath{{\mathbb R}}}
\newcommand{\KZ}{\ensuremath{{\mathbb Z}}}
\newcommand{\KN}{\ensuremath{{\mathbb N}}}
\newcommand{\KCP}{\ensuremath{{\mathbb C}{\mathbb P}}}
\newcommand{\KT}{\ensuremath{{\mathbb T}}}

\newcommand{\T}{\ensuremath{{\mathbb T}}}
\newcommand{\into}{\ensuremath{\hookrightarrow}}
\newcommand{\Aut}{\operatorname{Aut}}
\newcommand{\Ad}{\operatorname{Ad}}

\newcommand{\Br}{\operatorname{Br}}
\newcommand{\Dyn}{\operatorname{Dyn}}
\newcommand{\im}{\operatorname{im}}

\newcommand{\Cpct}{\ensuremath{\mathcal{K}}}
\newcommand{\CSet}{\ensuremath{\mathcal{S} \mathcal{E} \mathcal{T}}}
\newcommand{\CrPr}[3]{ \ensuremath{#1 \underset{#3}{\rtimes} #2} }
\newtheorem{theorem}{Theorem}[section] 
\newtheorem{definition}{Def}[section] 
\newtheorem{lemma}{Lemma}[section] 
\newtheorem{corollary}{Corollary}[section]

\begin{document} 
\title[Topological T-duality and Automorphisms]{Topological T-duality, 
Automorphisms and Classifying Spaces}
\author{Ashwin S. Pande}
\email{{ ashwin@hri.res.in} }
\date{\today}
\begin{abstract}
We extend the formalism of Topological T-duality to spaces
which are the total space of a principal $S^1$-bundle $p:E \to W$
with an $H$-flux in $H^3(E,\KZ)$ together 
the together with an automorphism of the continuous-trace algebra on $E$
determined by $H$.
The automorphism is a `topological approximation' to
a gerby gauge transformation of spacetime. 
We motivate this physically from Buscher's Rules for T-duality.
Using the Equivariant Brauer Group, we connect this problem to 
the $C^{\ast}$-algebraic
formalism of Topological T-duality of Mathai and Rosenberg \cite{MRCMP}.

We show that the study of this problem leads to the study of a purely 
topological problem, namely, Topological T-duality of
triples $(p,b,H)$ consisting of isomorphism classes of a principal
circle bundle $p:X \to B$  and classes $b \in H^2(X,\KZ)$ and 
$H \in H^3(X,\KZ).$ 
We construct a classifying space $R_{3,2}$ for triples 
in a manner similar to the work of Bunke and Schick \cite{Bunke}. 
We characterize $R_{3,2}$ up to homotopy and study some of its properties.
We show that it possesses a natural self-map which induces
T-duality for triples. We study some properties of this map. 
\end{abstract}
\maketitle

\section{Introduction \label{SecBr}}
Topological T-duality is an attempt to
study the T-duality symmetry of Type II String Theory \cite{Alvarez} 
using methods from Noncommutative and Algebraic 
Topology \cite{MRCMP, Bunke, BouEvMa, Daenzer}. 
In the simplest case, T-duality states that Type II A string theory 
on a certain background is equivalent to Type II B string theory on
another background, that is, acting on the original string theory by a 
canonical transformation (termed T-duality) transforms that theory
into the dual one\footnote{See Ref.\ \cite{Alvarez} item (6) on page (4) 
and Sec.\ (7) for a proof.}. 
For this to be possible, the background 
spacetime must carry a torus action (which need not be free). 
In this paper we only consider backgrounds which are principal circle bundles.

Both Type IIA and Type IIB String Theory backgrounds possess 
three massless bosonic fields: A graviton (associated to the metric),
an $H$-flux (associated to the Kalb-Ramond field $B$, $H=dB$) 
and a dilaton.  As is well known \cite{Minasian}, the $H$-flux is a 
closed integral three-form which is the gerbe curvature of a gerbe 
with connection form $B$ on that background. (The word `gerbe'
is used in the sense of Ref.\ \cite{Minasian}.)
It is a remarkable fact that the Topology and $H$-flux of the T-dual 
spacetime depend only on the Topology and $H$-flux of the original 
spacetime \footnote{See Ref.\ \cite{MRCMP, Bunke, BouEvMa, Daenzer} 
for a variety of approaches.}.  
This phenomenon is called Topological T-duality.

Let $X^m$ be a $m$-dimensional manifold which is a
principal circle bundle over a base $W$. Let $X^m \times Y^{10-m}$ be 
the manifold which is a model for the spacetime background. 
Suppose that there is a type II String Theory on this spacetime background.
We perform a T-duality along the circle orbits in $X^m$. 
Topological T-duality (see Refs.\ \cite{MRCMP, Bunke, BouEvMa, Daenzer})
claims that the underlying topological space of the physical T-dual 
background is $(X^{m})^{\#} \times Y^{10-m}:$
Here, $(X^{m})^{\#}$ is the Topological T-dual of the underlying topological
space of $X^{m}.$ 

Consider a spacetime with a $H$-flux of strength $H.$ We model the
spacetime by a manifold $X$ which is the total space of a principal
circle bundle $p:X \to W$ and the $H$-flux by a gerbe with connection 
on $X$ whose gerbe curvature form is $H.$ Let $V_{\alpha}$ be an open cover 
of $W = X/S^1$ and $U_{\alpha}$ be the lift of the cover to $X.$ 

For the convenience of the reader, we give the definition of a gerbe
and gerbe connection here. We follow the treatment in Minasian, 
Ref.\ \cite{Minasian}, Sec.\ (2.2) here\footnote{See also the treatment
in Chatterjee Ref.\ \cite{Chatterjee}}. Let $X$ and $\{ U_{\alpha} \}$ be
as above. For any $\alpha_1,\ldots,\alpha_n,$ let 
$U_{{\alpha_1}\ldots{\alpha_n}} = U_{\alpha_1} \cap \cdots \cap U_{\alpha_n}.$
\begin{definition}
A gerbe on $X$ is defined by the following data: 
\leavevmode
\begin{enumerate}
\item A line bundle $L_{\alpha \beta}$ on each two-fold intersection 
$U_{\alpha \beta}.$
\item An isomorphism $L_{\alpha \beta} \simeq L_{\beta \alpha}.$
\item A smooth nowhere-zero section
$f_{\alpha \beta \gamma}:X_{\alpha \beta \gamma}
\to \KC^{\ast}$ of the line bundle $L_{\alpha \beta}\otimes L_{\beta \gamma}
\otimes L_{\gamma \alpha}$ on each three-fold intersection
$U_{\alpha \beta \gamma}.$
\item $f_{\alpha \beta \gamma}$ satisfies the cocycle condition
$(\delta f)_{\alpha \beta \gamma \delta} = f_{\alpha \beta \gamma}
f^{-1}_{\beta \gamma \delta} f_{\gamma \delta \alpha} f^{-1}_{\delta
\alpha \beta} = 1$ on each four-fold intersection 
$U_{\alpha \beta \gamma \delta}.$
\end{enumerate}
\end{definition}

We now define a gerbe with connection also following Minasian:
\begin{definition}
A gerbe with connection on $X$ is a gerbe on $X$ together with a
connection $A_{\alpha \beta}$ on the line bundle $L_{\alpha \beta}$ in
each $U_{\alpha \beta}$ such that the section $f_{\alpha \beta \gamma}$
is covariantly constant with respect to the induced connection
on $L_{\alpha \beta}\otimes L_{\beta \gamma} \otimes L_{\gamma \alpha}:$
$$
A_{\alpha \beta} + A_{\beta \gamma} + A_{\gamma \alpha} =
\frac{1}{2\pi i} f^{-1}_{\alpha \beta \gamma} df_{\alpha \beta \gamma},
$$
and a two-form (the gerbe connection) $B_{\alpha} \in \Omega^2(U_{\alpha})$
such that $B_{\alpha} - B_{\beta} = dA_{\alpha \beta}$ on $U_{\alpha \beta}.$
\end{definition}

It is clear that $dB_{\alpha}=dB_{\beta}$ and hence the forms $dB_{\alpha}$ 
glue into a global three-form $H$ termed the gerbe curvature.
Physically, it models the field strength of the $B$-field.
This three-form is integral and defines a characteristic class 
$[H]$ in de Rham cohomology.

We now define the notion of a gauge transformation and large gauge
transformation of a gerbe. We 
follow Ref.\ \cite{Chatterjee} for the definition of a gauge transformation 
of a gerbe ( see Ref.\ \cite{Chatterjee}: 
A $1$-gauge transformation\footnote{ See Ref.\ \cite{Chatterjee}, 
Def.\ (2.2.5).} there is termed a gauge transformation here.
A $0$-gauge transformation\footnote{ See Ref.\ \cite{Chatterjee}, 
Def.\ (2.1.5).} there would be a family of automorphisms of each 
of the $L_{\alpha \beta},$ that is it would correspond to performing an
independent gauge transformation on each of the $L_{\alpha \beta}$
) and Ref.\ \cite{FMS} for the definition of a large
gauge transformation of the $B$-field ( see Ref.\ \cite{FMS}, pg. (11) before
Eq. (2.17)):
\begin{definition}
Let $X$ and $\{ U_{\alpha} \}$ be as above.
Let $U_{\alpha} \mapsto G_{\alpha}$ be an assignment of one-forms to the
open sets in the chart on $X.$
A gauge transformation of the gerbe on $X$ is the following transformation
of the $B$-field and its gauge field:
\leavevmode
\begin{itemize}
\item $A_{\alpha \beta} \mapsto A_{\alpha \beta} + 
(G_{\alpha}|_{U_{\alpha \beta}} - G_{\beta}|_{U_{\alpha \beta}})$,
\item $B_{\alpha} \mapsto B_{\alpha} + dG_{\alpha}.$
\end{itemize}
\end{definition}

In Ref.\ \cite{FMS} the authors 
point out\footnote{See pg.\ (11) before Eq.\ (2.17).}, 
that for a $l$-form field strength with $(l-1)$-form potential
$A,$ gauge transformations of the form $A \to A + \omega$ with
$[\omega]$ nontrivial are termed `large' gauge transformations.
For the $B$-field in Type II string theories, 
$l=3$ (see item no. (5) after Eq.\ (2.21) on page 13 
of Ref.\ \cite{FMS}). 

Hence for large gauge transformations, 
$dG_{\alpha}=\omega|_{U_{\alpha}},$ with $\omega$ a closed 
two-form on $X.$ Among these, there are gauge transformations for which
$\omega = F,$ with $F$ a closed integral two-form on $X.$
These transformations are special for a geometric reason: It is possible to 
tensor a gerbe with 
a line bundle $L:$ One sends $L_{\alpha \beta} \mapsto 
L|_{U_{\alpha \beta}} \otimes L_{\alpha \beta}, 
A_{\alpha \beta} \mapsto A|_{U_{\alpha \beta}} + A_{\alpha \beta},
B_{\alpha} \mapsto B + F|_{U_{\alpha}}$ where $A,F$ are the connection
and curvature forms of the connection on $L.$
(See Ref.\ \cite{Minasian} after Eq.\ (2.10).) 
It is these special gauge transformations that we consider in this
paper. Such a gauge transformation cannot be homotoped to zero 
without changing its cohomology class and is termed a large gauge 
transformation. These special transformations generate 
the gauge group. 

Consider a gerbe with connection on $X$ with gerbe curvature form $H$ which
is equivariant under the $S^1$-action on $X$: That is, for every
$t \in S^1,$ pulling the
gerbe connection back by the $S^1$-action map $\phi_t:X \to X$
causes the gerbe connection to undergo a gauge transformation.
Note that in this case the characteristic class $[H] \in H^3(X,\KZ)$ 
of the gerbe curvature form $H$ is {\em invariant} under pullback by
$\phi_t$ for every $t \in S^1$,
i.e., for every $t \in S^1,$ $[\phi_t^{\ast} H] = [H].$ 
(See Ref.\ \cite{Minasian}, before Eq.\ (2.3).) 

In this paper we study large gauge transformations
of the gerbe which are equivariant under the circle action on $X:$ We had noted above that a large gauge 
transformation naturally gives rise to a closed integral two-form $F$ on $X.$ 
Locally these transformations are of the form 
$B_{\alpha} \to B_{\alpha} + F|_{U_{\alpha}}$
where $F$ is the curvature of a line bundle $L$ naturally associated
to the large gauge transformation as described above. 
Such a gauge transformation has as a characteristic class the integral 
cohomology class of $F$, namely $[F] \in H^2(X,\KZ)$. 
We require that $L$ be equivariant under the circle action
on $X:$ That is, we require that the circle action $\phi_t$ on $X$ acts on $L$
by bundle automorphisms such that its curvature remains invariant 
under pullback by $\phi_t,$ 
$\phi_t^{\ast}(F) = F.$  It is clear that
(see Ref.\ \cite{Minasian} Appendix C for a treatment of
equivariant line bundles on a principal bundle in another context)
the Lie derivative of $F$ with respect to the
circle action vanishes. 

We now argue that for a smooth gerbe T-duality also T-dualizes
gerbe automorphisms. We use the description of T-duality in Type II string
theory using the theory of smooth gerbes given in the paper by Minasian 
(see Ref.\ \cite{Minasian}). After this, we will abstract this argument 
into a purely topological conjecture for automorphisms of gerbes on spaces 
and make a connection with the $C^{\ast}$-algebraic formalism of Topological
T-duality. 

Only for the purposes of this calculation, we use a different
notation to conform with the notation of Ref.\ \cite{Minasian}.
The notation will revert back to the previous notation as soon as this
calculation is over.
We let $\pi:X \to W$ be a smooth
principal circle bundle over a smooth base manifold $W.$
(See Ref.\ \cite{Minasian}, Sec.\ (2.1), note that $n=1$ in this paper,
and $M=W.$). On $X$ we assume given a metric and a $B$-field.
Let $\{U_{\alpha}\}$ be a cover of $W$ which we lift to a cover on $X.$
We choose coordinates $(\theta, x^{\alpha})$ locally 
on the cover $\{ U_{\alpha} \}$ where $\theta$ is the coordinate along
the circle direction. Greek subscripts on tensors now refer to the
$x^{\alpha}.$ (Hence, the connection form on $X$ is 
$\Theta = d\theta + A_{\alpha} dx^{\alpha}.$)

Let $\tilde{\pi}:X^{\#} \to W$ be the T-dual bundle.
We write the metric and $B$-field on $X$ in 
Kaluza-Klein form as
\begin{gather}
ds^2 = G_{00}(d \theta + A_{\alpha} dx^{\alpha})^2 + 
g_{\alpha \beta} dx^{\alpha} dx^{\beta} \nonumber \\
B = (d \theta + A_{\alpha}) \wedge B_{\beta} dx^{\beta} 
+ B_{\alpha \beta} dx^{\alpha} \wedge dx^{\beta}. \nonumber
\end{gather}
Here, $B_{\alpha}$ and $B_{\alpha \beta}$ are horizontal two-forms.

If we denote T-dual quantities by tilde superscripts, 
Buscher's rules take the form
\begin{gather}
\tilde{G}_{00} = 1/G_{00}, \nonumber \\
\tilde{A}_{\alpha} = B_{\alpha} \mbox{ , }
\tilde{B}_{\alpha} = A_{\alpha} ,\nonumber \\
\tilde{g}_{\alpha \beta} = g_{\alpha \beta} \mbox{ , }
\tilde{B}_{\alpha \beta} = B_{\alpha \beta}, \nonumber \\
\tilde{\Phi} = \Phi - \frac{1}{2} \ln(G_{00}). \nonumber
\end{gather}
(For example, see Ref.\ \cite{Kiritsis}, Ex.\ (6.12).)

Let $\Theta$ denote the connection form on the bundle
$X$ and $\tilde{\Theta}$ the connection form on the
T-dual bundle. From the above
\begin{align}
\Theta = d \theta + A_{\alpha} dx^{\alpha} \nonumber \\
\tilde{\Theta} = d \tilde{\theta} + B_{\alpha} dx^{\alpha} \nonumber
\end{align}
where we have used Buscher's Rules above.
Now, suppose $B \to B + F$ on $X$ where $dF=0$ and
$F$ is integral. 
We may write $F= \Theta \wedge F_1 + F_2$ where
$F_i$ are horizontal forms on $X$. The T-dual
$B$-field and metric will now be
\begin{eqnarray}
ds^2  =&&  G_{00}(d \tilde{\theta} + B_{\alpha} dx^{\alpha}+
F_{1\alpha} dx^{\alpha})^2 + 
g_{\alpha \beta} dx^{\alpha} dx^{\beta} \nonumber \\
\tilde{B} =&& (d \tilde{\theta} + B_{\alpha} dx^{\alpha} + 
F_{1\alpha} dx^{\alpha}) 
\wedge A_{\beta} dx^{\beta} + \nonumber \\
&& B_{\alpha \beta} dx^{\alpha} \wedge dx^{\beta}  + 
F_{2\alpha \beta} dx^{\alpha} \wedge dx^{\beta}. \nonumber
\end{eqnarray}
It is easy to see that since $dF=0$, $dF_1=0$ and, 
on the T-dual the form $dF_1$ is zero as well. 
However $F_2$ is not a closed form in
general.  

It can be seen directly from the definition
that $F_1 = p_!(F),$ and so, since $F$ is integral, 
$F_1$ is integral as well. 

Since $dF_1=0$, a coordinate transformation of $X^{\#}$ should be able
to remove the term $(d{\tilde\theta} + F_{1\alpha} dx^{\alpha})$ at least
locally. Locally, such a transformation would only rotate the circle fiber of 
$X^{\#}.$ The topological nontriviality of $F_1$ (the fact that it is integral)
implies that this transformation would act in a topologically nontrivial
manner on the total space of the bundle. Further,  since the gerbe 
connection on $X$ and on the T-dual are both equivariant
under the circle action, making such a coordinate transformation
would cause the T-dual $B$-field to undergo a gauge
transformation. 

This can be seen directly as follows: 
It is clear that this coordinate transformation would be 
a bundle automorphism (denoted $\phi$) of $\tilde{\pi}:X^{\#} \to W$ and so
$\tilde{\pi} \circ \phi = \tilde{\pi}$. 
Recall that $\tilde{B}_{\alpha} - \tilde{B}_{\beta} = 
d \tilde{A}_{\alpha \beta}$
by definition. We have that 
$$\phi^{\ast}\tilde{B}_{\alpha} - \phi^{\ast}\tilde{B}_{\beta} = 
\phi^{\ast}(d \tilde{A}_{\alpha \beta}) = 
d \phi^{\ast} \tilde{A}_{\alpha \beta}.$$
However, by equivariance\footnote{See 
Ref.\ \cite{Minasian}, Eq.\ (2.14).}
$$ \tilde{A}_{\alpha \beta}= q^{\ast} \tilde{a}_{\alpha \beta}
+ \tilde{h}_{\alpha \beta} \Theta^{\#}$$
where $\tilde{a}_{\alpha \beta}, \tilde{h}_{\alpha \beta}$ 
are horizontal one and zero forms on $X^{\#}.$
(See Ref.\ \cite{Minasian}, Eq.\ (2.14, 2.21) and Cor.\ (2.1).) Remember
that for a circle bundle, $m_{\alpha \beta} = 0$ so $\tilde{h}_{\alpha \beta}$
are actually the transition functions of the T-dual bundle $X^{\#} \to W.$)

Therefore, 
$$\phi^{\ast}\tilde{B}_{\alpha} - \phi^{\ast}\tilde{B}_{\beta} = 
\phi^{\ast} \tilde{A}_{\alpha \beta}= 
\phi^{\ast} \tilde{a}_{\alpha \beta}
+ \phi^{\ast} \tilde{h}_{\alpha \beta} \Theta^{\#}.$$
(Here $\phi^{\ast} \Theta^{\#} = \Theta^{\#}$
since $\phi$ is a bundle automorphism.)
If the bundle transformation $\phi$ was topologically
nontrivial, $\phi^{\ast}$ would
act nontrivially on the transition functions $\tilde{h}_{\alpha \beta}.$ 
As a result, 
$\phi^{\ast}\tilde{B_{\alpha}}$ transforms nontrivially on changing
charts.

The fact that $F$ is topologically nontrivial implies that 
the T-dual $B$-field has undergone a large gauge transformation.

As we had said earlier, a large gauge 
transformation of a gerbe 
has a characteristic class in $H^2(X,\KZ)$. 
We  may pick a large gauge transformation 
and ask for the characteristic class of
the gerbe gauge transformation on the T-dual.

From now on, $X$ will refer only to a $CW$-complex 
which is a model for a spacetime background. 

Thus, it is natural to ask the following: Given a class in $H^2(X,\KZ)$ 
does Topological T-duality naturally give a class in 
$H^2(X^{\#},\KZ)?$ We argued above that this phenomenon 
occurs for smooth gerbes. We will prove in this paper
that it occurs in the $C^{\ast}$-algebraic formalism of Topological
T-duality of Mathai and Rosenberg \cite{MRCMP} and also in the 
classifying space formalism of Bunke and coworkers \cite{Bunke}.
We will also attempt to obtain some information about the T-dual class.

Ref.\ \cite{MRCMP} suggests that given a spacetime with H-flux which also has 
a free $S^1$-action, one should attempt to construct an exterior
equivalence class of $C^{\ast}$-dynamical systems. 
In the associated $C^{\ast}$-dynamical system, the continuous-trace algebra 
together with its associated $\KR$-action should be viewed as
a `topological approximation' to the smooth gerbe on spacetime together
with a circle action on it. In Ref.\ \cite{TTDTF} such a 
$C^{\ast}$-dynamical system was naturally constructed from the 
data of a smooth gerbe with connection on $X$ in a large variety of examples.  

Further, Ref.\ \cite{MRCMP} demonstrated that the effect of the 
T-duality transformation on spacetime topology was given by the crossed 
product\cite{Will} construction in $C^{\ast}$-algebra theory in 
all the examples examined: The topological space underlying the T-dual 
spacetime was always the spectrum of the crossed-product 
algebra $\CrPr{\A}{\KR}{\alpha}.$ Based on this, Ref.\ \cite{MRCMP}
argued that it was natural to associate to the
T-dual the $C^{\ast}$-dynamical system 
$(\CrPr{\A}{\KR}{\alpha}, \hat{\KR},\alpha^{\#}).$

We may use the definition of Ref.\ \cite{MRCMP} and take as a 
model for a spacetime $X$ with $H$-flux, a continuous-trace algebra $\A$ with 
spectrum $X$ together with a lift $\alpha$ of the $S^1$-action on $X$ to 
a $\KR$-action on $\A$ \cite{RaeRos}. For simplicity, we restrict
ourselves to $S^1$-actions which have no fixed points.
Let $\Cpct$ denote the compact operators on a fixed seperable infinite
dimensional Hilbert space. Since the circle action has no fixed points, 
the spacetime $X$ is a principal circle 
bundle $p:X \to W.$ Let $\A = CT(X,\delta),$
$\delta \in H^3(X,\KZ).$ It can be shown that a lift $\alpha$ of the 
$S^1$-action on $X$ to a $\KR$-action on $\A$ exists and is unique up
to exterior equivalence. (See Ref.\ \cite{RosCBMS} Lemma (7.5) for a proof.)

We may model a large gauge transformation of the gerbe on spacetime 
by a locally unitary automorphism $\phi: \A \to \A$ of the associated 
continuous-trace algebra $\A.$ It is well known that such automorphisms
are determined up to exterior equivalence as automorphisms by their 
Phillips-Raeburn obstruction \cite{Phil} in $H^2(X,\KZ).$ 

There is a natural way to construct such a $\phi$ from the data above:
We noted above that a large gauge transformation of a gerbe on $X$
was naturally associated to a line bundle $L$ on $X.$ 
By a theorem of Phillips and Raeburn since
$X = \hat{\A}$ and we have a locally trivial principal circle bundle 
(the circle bundle $P \to X$ associated to the line bundle $L \to X$ ), 
we can naturally obtain a locally unitary automorphism group
$\alpha:\KZ \to \A$ such that the principal bundle 
$p:(\CrPr{\A}{\KZ}{\alpha})^{\wedge} \to X$ is isomorphic to $P$.
(See Ref.\ \cite{Phil}, Thm.\ (3.8), and let $G = S^1,$ and take
$\phi$ to be the generator of the group $\hat{G} = \KZ.$)
We may take $\phi$ to be the generator of the $\KZ$-action $\alpha.$
From the proof of the above theorem, it is clear that the characteristic
class of $L$ (which is the characteristic class of 
the large gauge transformation) will be the Phillips-Raeburn obstruction
of the automorphism $\phi.$ However, such an automorphism will not
commute with the $\KR$-action on $\A.$ In general, finding such an
automorphism is not a trivial matter. 

In the case when $X$ has an $S^1$-action which lifts to a $\KR$-action
on $\A$, the Phillips-Raeburn invariant does not take
the commutation of the automorphism with the $\KR$-action into account. 
We introduce a treatment based on the Equivariant Brauer Group in 
Sec.\ \ref{SecModel} below which fixes this problem. 

See Ref.\ \cite{Chase}, Sec.\ (2.3)
for an example of what would be needed to make such an automorphism
commute with the lift of the translation action on $X:$ If we take
$G=\KR$ there, the proof shows a way to obtain a $\KR$-action on 
$(\CrPr{\A}{\KZ}{\alpha})^{\wedge}$ making the
latter into a $\KR$-equivariant $S^1$-bundle over $\hat{\A}.$

In this paper we adopt a different approach to this problem by
an argument involving the Equivariant Brauer group. 

Recall that given a $S^1$-action $\psi$ on $X,$ there exists a lift
of $\psi$ to a $\KR$-action $\alpha$ on $\A$ which is unique up to
exterior equivalence such that $\alpha$ induces
the action $\psi$ on $X=\hat{\A}$ (See Ref.\ \cite{RosCBMS}, Lemma (7.5) 
for a proof.) 
If such a lift has been done, we may then ask if a given class in 
$H^2(X,\KZ)$ lifts to a unique automorphism
of $\A$ which commutes with the $\KR$-action on $\A.$ 
It is possible to prove that a lift exists. However, the lift will
in general be non-unique.
The following lemma is a slight generalization of Thm.\ (3.1) of 
Ref.\ \cite{ATMP}:
\begin{lemma}
Let $X$ be a principal bundle $p:X \to W.$ Let $\A = CT(X,\delta)$
for any $\delta \in H^3(X,\KZ).$ Let $\alpha_t$ be a lift of the $S^1$-action 
on $X$ to an $\KR$-action on $\A.$
\begin{enumerate}
\item Let $[\lambda] \in H^2(X,\KZ).$ Then, there is a $\KR$-action $\beta$
on $\A$ exterior equivalent to $\alpha$ and a spectrum-fixing $\KZ$-action
$\lambda$ on $\A$ which has Phillips-Raeburn obstruction $[\lambda]$
such that $\beta$ and $\lambda$ commute.
\item With the notation above, the action $\lambda$ induces a 
$\KZ$-action $\tilde{\lambda}$ on $\CrPr{\A}{\KR}{\alpha}.$ The induced
action on the crossed product is locally unitary on the spectrum
of the crossed product and is thus spectrum fixing.
\end{enumerate}
\label{LemATMP}
\end{lemma} 
\begin{proof}
The proof of Thm.\ (3.1) of Ref.\ \cite{ATMP} only uses the fact that 
$\A$ is a continuous-trace algebra. (See Ref.\ \cite{ATMP} and references 
therein.) The fact that $\A$ is $C_0(X,\Cpct)$ is not used anywhere in 
that theorem. Thus it applies even when the Dixmier-Douady invariant is nonzero.
\end{proof}

As a model for a gauge transformation of the gerbe associated to $\A$ 
we take a spectrum-fixing $C^{\ast}$-algebra automorphism $\phi$ on $\A$ which 
commutes with $\alpha$. We consider the $C^{\ast}$-dynamical system 
$(\A,\alpha \times \phi, \KR \times \KZ).$
The $C^{\ast}$-algebra automorphism $\phi$ then defines an automorphism 
$\psi$ of the crossed-product
algebra. (See Ref.\ \cite{ATMP} and Lemma (\ref{LemATMP}) above.)
What is the Phillips-Raeburn obstruction of this automorphism?
In this paper we obtain a solution to this problem using
a classifying space argument. (See Thm.\ (6.3), Part (2).)

Thus, given any class $[\lambda] \in H^2(X; \KZ)$, there exists a 
(not necessarily unique) action $\alpha$ of $\KR$ on $\A$ inducing the 
given action of $S^1=\KR/\KZ $ on $W$ and a commuting action $\lambda$ of
$\KZ$ on $\A$ with Phillips-Raeburn obstruction
$[\lambda].$ Also $\lambda$ passes to a locally unitary action
on ${\mathcal E} \simeq \CrPr{\A}{\KR}{\alpha}.$ The
actions $\alpha$ and $\lambda$ are (individually) unique up to
exterior equivalence, but unfortunately the pair
$(\alpha,\lambda)$, as an action of $\KR \times \KZ$, is 
not necessarily unique as a $C^{\ast}$-dynamical system, 
so this construction is not entirely
canonical. However, this non-uniqueness does {\bf not} change the 
Phillips-Raeburn invariant of the T-dual automorphism. Thus the
question asked above in the paragraph 
after Lemma (\ref{LemATMP}) is well-defined.

In Section (\ref{SecModel}), we introduce a point of
view based on the equivariant Brauer group, which measures 
precisely this lack of canonicity described above. 
We show that there is a natural\footnote{This does not contradict 
the statement in the previous paragraph,
as there we were trying to lift a class in $H^2(X,\KZ)$ to such a dynamical
system.} map, $T_R$, which sends a
$C^{\ast}$-dynamical system $(\A, \KR\times \KZ,\alpha \times \phi)$ to
a dual\footnote{We are using a different definition of dual dynamical 
system from Schneider's work, \cite{Schneider}, 
see the last part of Sec.\ (\ref{SecModel}).} dynamical system 
$(\CrPr{\A}{\KR}{\alpha}, \KR \times \KZ, \alpha^{\#} \times \phi^{\#})$.
To a $C^{\ast}$-dynamical system $(A,\KR \times \KZ, \alpha \times \phi)$ 
we asssociate a triple
\begin{center}
{\bf
 ($\hat{\A}$, Phillips-Raeburn invariant of $\phi$, 
Dixmier-Douady invariant of $\A$).
}
\end{center}
We argue that there is a well-defined map of triples, $T_{3,2}$, which 
sends a triple
\begin{center}
{\bf
(principal bundle $p:X \to B$,  $b \in H^2(X,\KZ)$, $\delta \in H^3(X,\KZ)$)
}
\end{center}
to another such triple which commutes with the map $T_R$ described above.

We then attempt to give an answer to the question raised in the paragraph 
after Lemma (\ref{LemATMP}) above.
In Sections (\ref{SecClass}-\ref{SecTriple}) we extend the formalism of 
Topological T-duality to triples of the type described above. We 
prove the existence of the map $T_{3,2}.$
Then we study some natural 
properties of the map $T_{3,2}.$ and show that
owing to the topological properties of this map, when 
the Dixmier-Douady invariant of $\A$  is fixed, there is a partition
of $H^2(X,\KZ)$ and $H^2(X^{\#},\KZ)$ into cosets such that the 
T-dual of an automorphism with Phillips-Raeburn invariant in a 
given coset  is an automorphism with Phillips-Raeburn invariant 
in another coset. This is enough to answer the question posed in the paragraph
after Lemma (\ref{LemATMP}) for several spaces.

\section{A model for the $H$-flux \label{SecModel}}
Let $X$ be a manifold which is a model for a spacetime with $H$-flux $H$ as in 
Sec.\ (1) and suppose $X$ is the total space of a circle bundle with 
$X/S^1 \simeq W.$ Let $V_{\alpha}$ be an open cover of $W$ and 
let $U_{\alpha}$ be the lift of the $V_{\alpha}$ to an open 
cover of $X.$ Then, the $H$-flux is the curvature of a smooth $S^1$-equivariant
gerbe on $X$ and $H=dB_{\alpha}$ locally, where $B_{\alpha}$ is 
the gerbe connection form in $U_{\alpha}.$
As discussed in Sec.\ (\ref{SecBr}),
changes in the $B$-field keeping $H$ fixed correspond to acting
on the gerbe on spacetime with a gerbe gauge transformation.
The gauge group of the gerbe is generated by the group of line bundles
on $X$ with connection. ( See Ref.\ \cite{Minasian}, Sec.\ (2) for example.)
That is, if $p:L \to X$ is a line bundle with curvature two-form $F,$
then, under such a gauge transformation, on each patch $U_{\alpha}$ 
we have $B'_{\alpha} = B_{\alpha} + F_{\alpha}$ where 
$F_{\alpha}=F|_{U_{\alpha}}$.
That is, $H=dB_{\alpha}$ locally, and after such a gauge transformation,
$H=dB'_{\alpha}$ with 
$d(B_{\alpha}-B'_{\alpha})=dF_{\alpha}= dF|_{U_{\alpha}}=0$ that is, 
for such gauge transformations, $(B_{\alpha}-B'_{\alpha})$ is closed. 
Note that $F$ is integral and hence so is $F_{\alpha}$. Therefore,
so is $(B_{\alpha}-B'_{\alpha})$. 
Since $X$ possesses a $S^1$-action (it is the total space of
a principal circle bundle with $W=X/S^1$), we require 
$L$ to be $S^1$-equivariant. 

As explained in the previous section, in the $C^{\ast}$-algebraic theory 
of Topological T-duality, the gerbe on spacetime is replaced by
a `topological approximation', a continuous-trace algebra $\A$ with
spectrum $X$ and Dixmier-Douady invariant equal to $H.$ It is well known that
spectrum-preserving automorphisms of $\A$ define a cohomology class 
(the Phillips-Raeburn invariant) in $H^2(X,\KZ).$ 
(See Ref.\ \cite{Phil}, also, see Ref.\ \cite{CKRW}, Lemma (4.4).)  

As described in Sec.\ (\ref{SecBr}) above, a gerbe gauge transformation of the 
type described above naturally defines an automorphism $\phi$ of the associated 
continuous-trace algebra with the Phillips-Raeburn class 
of the automorphism equal to the characteristic class of the line bundle 
$L$ above: See Ref.\ \cite{Phil} Thm.\ (3.8), take $E$ to be the circle bundle
associated to $L,$ $G = S^1$ and the 
automorphism $\phi$ to be the generator of the group $\hat{G} = \KZ.$ 
To study arbitrary changes in the $B$-field, that is, to study a 
general gerbe gauge transformation would probably require the introduction
of a smooth structure and would be difficult to do in the $C^{\ast}$-algebraic
picture of Topological T-duality. However, the theory of integral changes 
of the $B$-field is still quite interesting mathematically,
as the following sections show.

The above automorphism $\phi$ of $\A$ gives a $\KZ$-action 
on $\A.$  However, $\A$ possesses a $\KR$-action already 
which is a lift of the $S^1$-action on $X$ to $\A.$ Now, the gerbe
automorphism is equivariant under the $S^1$-action on $X,$ hence, $L \to X$ is
an equivariant line bundle. Therefore, we require that $\phi$ commute with 
the $\KR$-action on $\A.$ It is natural to ask if we obtain 
a $\KZ \times \KR$-action on $\A.$ The action obtained will, in general, 
depend on the exterior equivalence classes of the $\KR$-action on $\A$ 
and the $\KZ$-action induced by the automorphism $\phi$.  
The exterior equivalence class of $\phi$ is arbitrary,
since we obtained $\phi$ by lifting the characteristic class of $L$
to an spectrum-preserving automorphism of $\A$. Hence,
we need not obtain a unique $\KZ \times \KR$-action on $\A$.
However, this will not affect anything as the following shows.

It will be useful to recall the notion of the Equivariant
Brauer Group\footnote{We use Ref.\ \cite{CKRW} here.}: 
Let $X$ be a second countable, locally compact, Hausdorff
topological space and let $G$ be a second countable, locally compact,
Hausdorff topological group acting on $X$.
Following Ref.\ \cite{CKRW}, let $\mathfrak{B}\mathfrak{r}_{G}(X)$ denote the 
class of pairs $(A,\alpha)$ consisting of continuous-trace algebras $\A$ on 
$X=\hat{\A}$ together with a lift $\alpha$ of the $G$-action on $X$ to $\A.$
We say that the dynamical sytem  $(\A,\alpha)$ 
is equivalent to $(\B,\beta)$ if there
is a Morita equivalence bimodule $_{\A}M_{\B}$ together with a
strongly continuous $G$-action by linear transformations $\phi_s$ 
on $M$ such that for every $s \in G$,
$\alpha_s(\langle x,y \rangle_{\A}) = \langle \phi_s(x),\phi_s(y) 
\rangle_{\A}$,
and 
$\beta_s(\langle x,y \rangle_{\B}) = \langle \phi_s(x),
\phi_s(y) \rangle_{\B}$.
(Recall, the image of the inner product 
$\langle,\rangle_{\A}: M \times M \to \A$
is dense in $\A$ and similarly for $\langle,\rangle_{\B}$.)
It is shown in Ref.\ \cite{CKRW} that 
this is an equivalence relation and that 
the quotient is a {\em group}. This group is termed the
Equivariant Brauer Group $\Br_{G}(X)$ of $X$. The group operation 
is the $C_0(X)$-tensor product of continuous-trace algebras and 
group actions. 

The Equivariant Brauer Group is actually a special
case of the Brauer group of a groupoid applied to the transformation
groupoid $G \times X.$ It is shown in Ref.\ \cite{Kumjian} that the
Brauer group of a groupoid is isomorphic to a groupoid cohomology group. 
Hence the Equivariant Brauer Group is actually an abelian-group-valued 
functor contravariant in $G$ and in $X.$

Let $X$ be a space with a $S^1$-action $\psi$.
Suppose $\A$ was a continuous-trace $C^*$-algebra with spectrum $X$
which possessed an $\KR$-action $\alpha$ which induced the action
$\psi$ on $X=\hat{\A}$. Suppose we had a lift of a class in $H^2(X,\KZ)$ to a 
spectrum-preserving automorphism of $\A$ which commuted with this 
$\KR$-action on $X.$ We would obtain a $\KR \times \KZ$-action
on $\A$ and hence an element of $\Br_{\KR \times \KZ}(X)$. 
By Lemma (\ref{LemATMP}) above, 
there is an induced $\KZ$-action on $\CrPr{\A}{\KR}{\alpha}.$ 
Now, the lift of a class in $H^2(X,\KZ)$ to a $\KR \times \KZ$-action on
$\A$ need not be unique. Hence, the induced action on the crossed product
is not unique. However,we really only care about the restriction of this action 
to the $\KZ$-factor up to exterior equivalence, that is we wish to
determine the Phillips-Raeburn invariant of the induced $\KZ$-action on the 
crossed product $C^{\ast}$-algebra.
We outline a way to approach this problem in this Section. 
In order to quantify this lack of uniqueness we now 
study $\Br_{\KR \times \KZ}(X)$ in more detail. 

In Ref.\ \cite{CKRW}, the authors
show (see Ref.\ \cite{CKRW}, Thm.\ (5.1)) that there is a natural map
$\xi:H^2_M(\KR \times \KZ, C(X,\T)) \to \Br_{\KR \times \KZ}(X),$ where
$H^2_M$ is the Moore cohomology group mentioned in Ref.\ \cite{CKRW}.
It is really only the class of the above $\KR \times \KZ$-action in
$\Br_{\KR \times \KZ}(X)/\im(\xi)$ that we need
to determine the Phillips-Raeburn invariant of the induced $\KZ$-action
on the crossed product $C^{\ast}$-algebra. 
We compute this group in Lemma (\ref{LemRXZCD}) 
below.

Note that elements of $\Br_{\KR \times \KZ}(X)$ would 
consist of a Morita equivalence
bimodule between $\KR \times \KZ$-dynamical systems and an $\KR$-equivariant
automorphism of the module compatible with the $\KZ \times \KR$-action
on the dynamical systems. Thus, by studying $\Br_{\KR \times \KZ}(X)$
we are simply adding more structure to
$\Br_{\KR}(X).$ 

We conjecture that elements of 
this group are a good model for spacetimes with a possibly nonzero $B$-field. 
Given a dynamical system $(\A, \alpha \times \phi, \KR \times \KZ)$
in $\Br_{\KR \times \KZ}(X),$ we may forget the $\KZ$-action to obtain a 
$\KR$-dynamical system and we may forget the $\KR$-action to obtain a 
$\KZ$-dynamical system. Hence we obtain forgetful maps 
$F_1:\Br_{\KR \times \KZ}(X) \to \Br_{\KZ}(X)$ 
and $F:\Br_{\KR \times \KZ}(X) \to \Br_{\KR}(X).$  
We formalize the above discussion in the following 
\begin{definition}
Let $X$ be a locally compact, finite dimensional CW-complex homotopy
equivalent to a finite CW-complex. Let $X$ also be a free
$S^1$-space with $W=X/S^1,$ so that we have a principal $S^1$-bundle
$p:X \to W.$ An element\footnote{ Here $\alpha$ is a lift of the 
$S^1$-action on $X$ to a $\KR$-action on $\A,$ while 
$\phi$ is a commuting spectrum-fixing $\KZ$-action on $\A.$}
$y = [ \A, \alpha \times \phi]$ of
$\Br_{\KR \times \KZ}(X)$ is now defined to be a model for a space
with nonzero $H$-flux or zero $H$-flux and a nonzero integral $B$-field. 
The $H$-flux $H$ is\footnote{Recall that the forgetful map
$\Br_{\KR}(X) \to \Br(X) \simeq H^3(X,\KZ)$ is an isomorphism
\cite{CKRW}.}
$H=F(y) = [\A]$. If $H=0,$ the $B$-field is the unique 
class\footnote{Note that 
$\Br_{\KZ}(X) \simeq H^3(X,\KZ) \oplus H^2(X,\KZ)$.}
in $H^2(X,\KZ)$ which determines $F_1(y)$.
This class is equal to the Phillips-Raeburn class of $\phi$.
\label{DefSetup}
\end{definition}

By Ref.\ \cite{CKRW}, there is a natural filtration of 
$\Br_{\KR \times \KZ}(X)$ given by 
$0 < B_1 < \ker(F) < \Br_{\KR \times \KZ}(X),$
where $B_1$ is the quotient $H^2_M(\KR \times \KZ, C(X,\T))/\im(d_2^{''})$.
We argue below that each step in this filtration corresponds
to one of the gauge fields in the problem.
 
We need the following characterization of the groups and maps in the above
filtration. (Note that $B_1 = \ker(\eta)$ is determined in 
Lemma (\ref{LemRXZCD}) below).
\begin{theorem}
Let $p:X \to W$ be as above.
\begin{enumerate}
\item We have a split short exact sequence
$$
0 \to \ker(F) \to \Br_{\KR \times \KZ}(X) \overset{F}{\to} \Br(X) \cong
H^3(X,\KZ) \to 0
$$
where $F:\Br_{\KR \times \KZ}(X) \to \Br(X)$ is the forgetful map.
\item We have a surjective map $\eta:\ker(F) \to H^2(X,\KZ).$
\item We have a natural isomorphism $H^1_M(\KR, C(X,\T)_0) \simeq C(W,\KR).$
\item The group $H^2_M(\KR \times \KZ, C(X,\T))$ is connected
and there is a natural surjective map 
$q: H^2_M(\KR \times \KZ,C(X,\T)) \to C(W,\T)_0.$
\end{enumerate}
\label{ThmBrauer}
\end{theorem}
\begin{proof}
\begin{enumerate}
\item We have a forgetful homomorphism 
$F_1:\Br_{\KR \times \KZ}(X) \to \Br_{\KR}(X),$ where 
$F_1:[\A,\alpha \times \phi] \to [\A,\alpha]$. This map is
obviously surjective, since we have a section 
$s:[\A,\alpha] \to [\A, \alpha \times \operatorname{id}]$.

Since $\Br_{\KR}(X) = H^3(X,\KZ)$ (by Sec.\ (6.1) of Ref.\ \cite{CKRW}),
the kernel of $F_1$ consists of Morita equivalence classes of
dynamical systems $[\A,\alpha \times \phi]$ such that $\delta(\A) = 0$.
Thus, it actually consists of the group $\ker(F)$, where 
$F:\Br_{\KR \times \KZ}(X) \to \Br(X)$ is the map forgetting the
group action.

\item By Thm.\ (5.1) of Ref.\ \cite{CKRW}, we have a homomorphism
$\eta:\ker(F) \to H^1_M(\KR \times \KZ, H^2(X,\KZ))$. Now, by
Thm.\ (4.2) of Ref.\ \cite{MRCMP}, we have that 
$H^1_M(\KR \times \KZ,M) \simeq H^1(B(\KR \times \KZ),M)$ for
any discrete $\KR \times \KZ$ module $M$. Also, 
$B(\KR \times \KZ) \simeq S^1$, so 
$H^1_M(\KR \times \KZ, H^2(X,\KZ)) \simeq H^2(X,\KZ)$.

By Thm.\ (5.1) item (2) of \cite{CKRW}, the image of $\eta$ has
range which is all of $H^2(X,\KZ)$ since 
$H^3_M(\KR \times \KZ, C(X,S^1)) = 0$ by Thm.\ (3.1) of Ref.\ \cite{ATMP}. 
Hence we have a surjective homomorphism $\eta: \ker(F) \to H^2(X,\KZ)$.

\item We have the following short exact sequence of $\KR$-modules
$$
0 \to H^0(X,\KZ) \to C(X,\KR) \to C(X,\T)_0 \to 0
$$
From the associated long exact sequence for $H^{\ast}_M$, we find
that $H^i_M(\KR, H^0(X,\KZ)) \simeq 0$, $i = 1,2$ by Cor.\ (4.3)
of Ref.\ \cite{MRCMP}; hence 
$H^1_M(\KR, C(X,\T)_0) \simeq H^1_M(\KR, C(X,\KR))$.
By Thms.\ (4.5,4.6,4.7) of Ref.\ \cite{MRCMP}, 
$H^1_M(\KR,C(X,\KR)) \simeq H^1_{\text{Lie}}(\KR,C(X,\KR)_{\infty})$.
Here $C(X,\KR)_{\infty}$ are the $C^{\infty}$-vectors for the
$\KR$-action on $C(X,\KR)$ and so are the functions which are
smooth along the $S^1$-orbits.

The complex computing the Lie algebra cohomology of $\KR$ shows that
this group is exactly the functions in $C(X,\KR)_{\infty}$ modulo
derivatives of functions in $C(X,\KR)_{\infty}$ by the generator
of the $\KR$-action.

This group is isomorphic to $C(W,\KR)$ via the `averaging' map
$f \to \int_{S^1} \phi_t \circ f dt$ where $\phi_t \circ f$ is
$f$ shifted by the $S^1$-action on $X.$

\item  We use the spectral sequence calculation of Ref.\ \cite{ATMP},
Thm.\ (3.1) to
note that this group is isomorphic to $H^1_M(\KR,H^1_M(\KZ,C(X,\T))).$
Since $\KZ$ is discrete and acts trivially on $C(X,\T),$ we
have $H^1_M(\KZ,C(X,\T)) \simeq H^1(\KZ,C(X,\T)) \simeq C(X,\T).$
Hence we need to calculate $H^1_M(\KR,C(X,\T)).$

We have the following short exact sequence of $\KR$-modules
$$
0 \to C(X,\T)_0 \to C(X,\T) \to H^1(X,\KZ) \to 0
$$
where $C(X,\T)_0$ is the connected component of $C(X,\T)$ containing
the constant maps.

This gives us a long exact sequence
\begin{gather}
H^0_M(\KR,C(X,\T)) \to H^0_M(\KR,H^1(X,\KZ)) \to H^1_M(\KR,C(X,\T)_0) 
\to \nonumber \\
H^1_M(\KR,C(X,\T)) \to H^1_M(\KR,H^1(X,\KZ)) \to \ldots \nonumber \\
\end{gather}

Also, by Cor.\ (4.3) of Ref.\ \cite{MRCMP}, we have that 
$H^1_M(\KR,H^1(X,\KZ)) \simeq 0.$ Again, by Cor.\ (4.3) of
Ref.\ \cite{MRCMP}, we find that $H^0_M(\KR,H^1(X,\KZ)) \simeq H^1(X,\KZ)$
(since $B\KR$ is contractible). Also $H^0_M(\KR,C(X,\T))$ consists
of the $\KR$-invariant functions in $C(X,\T)$ and hence is
naturally isomorphic to $C(W,\T).$

Hence we find an exact sequence
\begin{gather}
C(W,\T) \to H^1(X,\KZ) \to H^1_M(\KR,C(X,\T)_0) \to H^1_M(\KR,C(X,\T)) \to 0.
\label{EqH2MRXZ}
\end{gather}

The map $C(W,\T) \to H^1(X,\KZ)$ is the composite
$C(W,\T) \to H^1(W,\KZ) \stackrel{p^{\ast}}{\to} H^1(X,\KZ).$
Its cokernel is $H^1(X,\KZ)/p^{\ast}(H^1(W,\KZ))$ which is the
image of $p_!:H^1(X,\KZ) \to H^0(W,\KZ)$ by the Gysin sequence.
(The image can only be $0$ or $\KZ$ if $X$ is connected.) Using
the isomorphism mentioned in the previous item of this lemma,
we see that we need to find the connecting map 
$\im(p_!) \to H^1_M(\KR,C(X,\T)_0) \simeq C(W,\KR)$. This map
sends any class in $H^0(W,\KZ)$ to a constant $\KZ$-valued function
on $W.$

The above exact sequence now becomes
\begin{gather}
0 \to \im(p_!) \to C(W,\KR) \to H^1_M(\KR,C(X,\T)) \to 0.
\end{gather}

So $H^1_M(\KR,C(X,\T))$ is isomorphic to the quotient of 
$C(W,\KR)$ by $\im(p_!).$ It surjects onto the quotient of
$C(W,\KR)$ by all of $H^0(W,\KZ)$ which is isomorphic to 
$C(W,\T)_0.$
\end{enumerate}
\end{proof}

\begin{lemma}
We have a commutative diagram with exact rows and columns
\begin{equation}
\CD
 @. 0 @. @. @. \\
@. @VVV @. @. @. \\
@. H^2_M(\KR \times \KZ, C(X,\T))/\im(d_2^{''}) @. @. @. \\
@. @VVV @. @. @. \\
0 @>>> \ker(F) @>>> \Br_{\KR \times \KZ}(X) @>F>> \Br_{\KR}(X) @>>> 0 \\
@. @V{\eta}VV @. @. @. \\
@. H^2(X,\KZ) @. @. @. \\
@. @VVV @. @. @. \\
@. 0 @. @. @. \\
\endCD
\nonumber
\end{equation}
where $d_2^{''}$ is defined in Ref.\ \cite{CKRW}, Thm.\ (5.1).
The group $\ker(\eta) = H^2_M(\KR \times \KZ, C(X,\T))/\im(d_2^{''})$ is 
isomorphic to a quotient of $C(W,\KR).$ 
\label{LemRXZCD}
\end{lemma}
\begin{proof}
The vertical and horizontal short exact sequences above are of the
form $0 \to B_i \to B_{i+1} \to B_{i+1}/B_i \to 0$ where
the $B_i$ are the groups in the filtration of $\Br_{\KR \times \KZ}(W)$ 
described in the unnumbered Theorem on page (153) of Ref.\  \cite{CKRW}.

The horizontal short exact sequence is the sequence of 
part (1) of Thm.\ (\ref{ThmBrauer}). The map $\eta$ is the map
of part (2) of Thm.\ (\ref{ThmBrauer}). 
We are interested in the group $\ker(\eta).$ 
This is the collection of $C^{\ast}$-dynamical systems 
of the form $(C_0(X,\Cpct), \alpha)$ with $\alpha$ inner. (Note
that the Mackey obstruction of $\alpha$ may be nonzero.)
By Part (2) of Thm.\ (\ref{ThmBrauer}), above, the map $\eta$ in Thm.\ (5.1) of
Ref.\ \cite{CKRW} is the map $\eta:\ker(F) \to H^2(X,\KZ)$ above.
By part (3) of Thm.\ (5.1) in Ref.\ \cite{CKRW}, 
we have natural homomorphisms
$\xi:H^2_M(\KR \times \KZ,C(X,\T)) \to \Br_{\KR \times \KZ}(X),$ 
and  $d_2^{''}:H^2(X,\KZ) \to H^2_M(\KR \times \KZ, C(X,\T)).$
Also, from that theorem it follows that $\ker(\eta)$ is identical to 
$\im(\xi) \subseteq \Br_{\KR \times \KZ}(X)$ 
and\footnote{Note that the translation action on $X$ acts trivially
on $H^2(X,\KZ)$ and that the $\KR \times \KZ$-action covers this translation
action.} we have the identifications
$\im(\xi) \simeq H^2(\KR \times \KZ, C(T,\T))/\ker(\xi),$
$\ker(\xi) = \im(d_2^{''}).$
Hence, it follows that 
$\ker(\eta) = H^2_M( \KR \times \KZ, C(X, \T))/\im(d_2^{''}).$

We need to calculate this group. By the previous theorem
(in particular, Eq.\ (\ref{EqH2MRXZ}) above), $H^2_M(\KR \times \KZ,C(X,\T))$
is isomorphic to $H^1_M(\KR, C(X,\T))$ which is isomorphic to
$C(W,\KR)/\im(p_!).$ Therefore, $B_1$ should be isomorphic to
a quotient of $C(W,\KR).$ 

%
%

%
%

The maps $F,\eta,q$ were defined in the previous lemma.
\end{proof}

%

We now make the following dictionary
\begin{itemize}
\item $y \in \Br_{\KR \times \KZ}(X),$ $y$ not in $\ker(F)$ $\Leftrightarrow$
Space $X$ with $H \neq 0.$ Here, $H=F(y).$
\item Element $y \in \ker(F) \subseteq \Br_{\KR \times \KZ},$
$y$ not in $H^2_M(\KR \times \KZ, C(X,\T))$ $\Leftrightarrow$
Space $W$ with $H=0, B \neq 0.$ Here, $B = \eta(y).$
\item $y \in H^2_M(\KR \times \KZ, C(X,\T))/\im(d_2^{''}) \subseteq 
\Br_{\KR \times \KZ}(X)$ $\Leftrightarrow$ Space $X$ with $H=0,B=0,A \neq 0.$
We obtain an element of $C(W,\T),$
however, this may instead be viewed as a collection of elements 
in $C(X,\KR)$ which are constant on the $S^1$-orbits of $X.$ 
Any two of these elements differ by an element 
of $C(X,\KZ) \simeq H^0(X,\KZ).$ 

This should be compared with the allowed gauge transformations of the
$B$-field in string theory. 

\item $C_0(X,\Cpct)$ with the lift of the $\KR$-action and the trivial
$\KZ$-action $\Leftrightarrow$ Space $X$ with $H=0,B=0, A= 0$.
\end{itemize}

Note that the last item above is exactly the $C^{\ast}$-dynamical system
assigned to a space $X$ with zero $H$-flux in Ref.\ \cite{MRCMP}.

Suppose we had a string background which was a principal bundle
$X$ with $X/S^1 \simeq W$ together with a continuous-trace algebra $\A$ with
spectrum $X$ satisfying the conditions of Def.\ (\ref{DefSetup}) above.
By Def.\ (\ref{DefSetup}) we would obtain an element 
$y = [\A,\alpha \times \phi,\KR \times \KZ]$ 
in $\Br_{\KR \times \KZ}(X)$ (see Ref.\ \cite{CKRW}).
We remarked earlier that $\A$ should be viewed as 
a `topological approximation' to the gerbe on spacetime.
If $y \in \Br_{\KR \times \KZ}(X),$ we have $[H]=F(y).$
Then $F^{-1}([H])$ is the coset $y \circ \ker(F)$ and we note that
we may change $y$ by an element $x$ of $\ker(F)$ without changing the
$H$-flux.

That is, by definition of the
Brauer group (see Ref.\ \cite{CKRW} Prop. (3.3)), corresponding to the 
element $yx \in \Br_{\KR \times \KZ}(X),$ we obtain a new 
$C^{\ast}$-dynamical system 
$$[\A, \alpha \times \psi] \circ [\A,\alpha \times \phi] \simeq 
[\A \otimes_{C_0(X)} \A, (\alpha \times \psi)\circ(\alpha 
\times \phi)] \simeq [\A,\alpha \times (\psi \circ \phi )].$$
Here, the $\KR$-action obtained from the composite dynamical system 
$[\A, (\alpha \times (\phi \circ \psi)]$ is still the lift of the
circle action on $X,$ but the Phillips-Raeburn invariant of the $\KZ$-action 
has changed, it has shifted by $\eta(x).$ 

Given the above, if we act on the gerbe on spacetime by a gerbe 
gauge transformation $x,$ the Phillips-Raeburn invariant of 
the $\phi$ factor of $y=[\A, \alpha \times \phi]$ should be shifted by 
$\eta(x).$

We may view this action as affecting the $H$-flux on spacetime by shifting
the $B$-field, that is, if $H=dB$ locally, we act by a gerbe automorphism
to obtain $H=dB'$ locally.  As argued at the beginning of this section,
we are restricting ourselves to large gauge transformations of the $B$-field
so $(B-B')$ is actually a global quantity, in fact, it is a closed, integral
two-form on $F$ on $X.$ 

Hence, we only allow shifts of $B$ by
elements of $H^2(X,\KZ) \stackrel{i}{\hookrightarrow} H^2(X,\KR),$ 
where $i$ is the canonical inclusion.
Also, physically, $i\circ \eta(x)$ can only equal $[(B'-B)].$  
 
Similarly, if $y \in \ker(F),$ then $H=0, B=\eta(y)$ and
changing $y$ by an element $z$ of $B_1=H^2_M(\KR \times \KZ,C(X,\T))$ doesn't
change $B$ but might correspond to making a change in the $A$-field,
the gauge field of the $B$-field. 

We have argued here that it is natural to associate the group 
$\Br_{\KR \times \KZ}(X)$ to a spacetime $X$ with $H$-flux 
since it captures properties
of the $H$-flux on that spacetime. The assignment 
$X \to \Br_{\KR \times \KZ}(X)$ is a functor $\R$
on a certain category defined below (see paragraph before 
Lemma \ref{LemBrFilt} below). 

We now construct the functor $\R$ above. We also define several other
functors of interest which we will need later. 
In the rest of this section we will study these functors in more detail. 
For circle bundles, we obtain the functor of Bunke et al 
(see Ref.\ \cite{Bunke}) from the $C^{\ast}$-algebraic theory by
using a construction based on the Equivariant Brauer Group 
(before Lemma (\ref{LemTR}) below). We also explain how these constructions help
answer the question raised in the paragraph after Lemma (\ref{LemATMP})
above.

As was discussed before Def.\ (2.1) above,
the Equivariant Brauer Group of $X$ is a special case of the Brauer
Group of a groupoid applied to the transformation groupoid 
$G \times X.$ It is shown in Ref.\ \cite{Kumjian} 
that the Brauer Group of a groupoid is isomorphic to 
a groupoid cohomology group. Hence, the Equivariant Brauer Group is 
an abelian-group-valued functor contravariant in $G$ and $X.$ 
The following lemma is then obvious:
\begin{lemma}
Let $Y \to Z$ be a principal bundle of $CW$-complexes.
Let $f: W \to Z$ be a continuous map. Let $X \to W$ be the pullback
bundle and $\phi:X \to Y$ be the map induced by pullback. 
Consider $\Br_{\KR \times \KZ}(X)$ as described above.
Then, pullback of $C^{\ast}$-algebras induces the following maps
\leavevmode
\begin{enumerate}
\item $\lambda:\Br_{\KR}(Y) \to \Br_{\KR}(X)$. 
\item $\kappa:\Br_{\KR \times \KZ}(Y) \to \Br_{\KR \times \KZ}(X)$.
\end{enumerate}
\label{LemBrRRZ}
\end{lemma}

Consider the category
of principal $S^1$-bundles over $CW$-complexes with morphisms pullback maps.
Let $X \to W$ be a principal $S^1$-bundle over a $CW$-complex
$W$. Let $(\A,\alpha,\KR \times \KZ)$ be a $\KR \times \KZ$
$C^{\ast}$-dynamical system with spectrum $X$ with $\alpha$ a lift
 of the $S^1$-action on $X$ to a $\KR \times \KZ$ action on $\A$ as 
described in the previous paragraph. If $Y \to Z$ is another such
principal $S^1$-bundle, then $(\A,\alpha,\KR \times \KZ)$ pulls back 
to a $C^{\ast}$-dynamical system on $Y$. This induces a pullback morphism
$\phi:\Br_{\KR \times \KZ}(X) \to \Br_{\KR \times \KZ}(Y)$.
Under pullback, $\Br_{\KR \times \KZ}$ gives an abelian-group valued 
functor on the category of principal $S^1$-bundles of $CW$-complexes
with arrows pullback squares of bundles:
Given a pullback square 
\begin{gather}
\begin{CD}
Y @>\tilde{f}>> X\\
@VVV  @VVV \\
Z @>f>> W
\end{CD}
\end{gather}
the functor assigns $\Br_{\KR \times \KZ}(X)$ to the
any object $P \to W$ and the pullback map $\phi$ to the above
pullback square.

\begin{lemma}
Let $X,Y,Z,W$ be as described in the previous paragraph. Let $\phi$
be as above. The map $\phi$ preserves the groups in the natural filtration on 
$\Br_{\KR \times \KZ}(X)$ in Ref.\ \cite{CKRW}. 
\label{LemBrFilt}
\end{lemma}
\begin{proof}
The natural filtration on $\Br_{\KR \times \KZ}(X)$ is 
$$H^2(\KR \times \KZ, C(X,\T))/\im(d_2^{''}) \subseteq \ker(F) \subseteq 
\Br_{\KR \times \KZ}(X)$$ (see Lemma (\ref{LemRXZCD}) above). Under pullback,
as shown above, $\Br_{\KR \times \KZ}$ maps naturally. Under pullback,
the Dixmier-Douady invariant also pulls back, so, if an element 
in $\ker(F) \subseteq \Br_{\KR \times \KZ}(Y)$ corresponded to a 
dynamical system of the form 
$(C_0(Y,\Cpct),\alpha \times \beta,\KR \times \KZ)$ its pullback 
would be of the form $(C_0(X,\Cpct),\alpha' \times \beta',\KR \times \KZ)$.
Also, a locally unitary map pulls back to another locally unitary map. 
Thus, under pullback an element of $\ker(F)$ maps to another element
of $\ker(F)$. By Lemma (\ref{LemRXZCD}) above, 
$H^2(\KR \times \KZ, C(X,\T))/\im(d_2^{''})$
is a quotient of $C(W,\KR)$ and pullback of continous-trace algebras
induces a natural map to
$H^2(\KR \times \KZ, C(Y,\T))/\im(d_2^{''})$ 
(which is a quotient of $C(Z,\KR)$). 
\end{proof}

Let $X \to W$ be a principal circle bundle as described just before
Lemma (\ref{LemBrFilt}) above. For any $i,$ let $G_i$ be 
any of the following contravariant abelian-group-valued 
functors\footnote{ By naturality, (See Ref.\ \cite{CKRW}), 
these are the groups that would appear in the discussion of 
$\KR \times \KZ$-actions on $CT(X,\delta)$ which on $X$ cover the $S^1$-action
on $X$. } which assign to $X$ the value $\Br_{\KR}(X),\Br_{\KR \times \KZ}(X)$,
any of the groups in the filtration on $\Br_{\KR \times \KZ}(X)$, any
of the cohomology groups $H^{q}(X,\KZ), H^{q}_M(\KR \times \KZ,C(X,\T))$ or 
$H^p(X,\KZ)\oplus H^q(X,\KZ)$. Let 
$N:G_1 \to G_2$ be a natural transformation between the $G_i$.
\begin{theorem}
Let $W$ be a $CW$-complex. Choose a principal $S^1$-bundle $E_p$ on $W$
for every $[p] \in H^2(W,\KZ)$. Let $G,G_i$ be as described before
this theorem and let $N:G_1 \to G_2$ be a
natural transformation as described above. 
Then we have the following
\leavevmode
\begin{enumerate}
\item The group $\Aut(E_p)$ acts on $G(E_p)$ for every $E_p$. 
\item For any of the choices of $G$ listed before this theorem, let $Y_p$ 
denote the set $G(E_p)/\Aut(E_p)$. Let 
$P_G(W) = \amalg_{[p] \in H^2(W,\KZ)} Y_p$. 
Then, the set $P_G(W)$ doesn't depend on the choice of $E_p$.
\item  Let {\bf Set} be the category whose objects are sets and whose
morphisms are functions. $P_G(W)$ is a contravariant {\bf Set}-valued 
functor\footnote{See also the definitions of $\mbox{Par}(B)$ after 
Rem.\ (2.2) and $\Dyn(E,B)$ after Prop.\ (2.8) in Ref.\ \cite{Schneider}}
on $CW$-complexes for every choice of $G$.
\item The natural transformation $N:G_1 \to G_2$ yields a natural 
transformation $P_N:P_{G_1} \to P_{G_2}$.
\item Let $N:G_1 \to G_2$ be the natural transformation mentioned above.
If the map $N(E_p):G_1(E_p) \to G_2(E_p)$ is always surjective for every 
principal bundle $E_p \to W$ for every $CW$-complex $W,$ 
then the induced map $P_N(W):P_{G_1}(W) \to P_{G_2}(W)$ is
surjective for every $CW$-complex $W.$
Similarly if $N(E_p)$ is injective or bijective then so is $P_N(W)$ for
every $CW$-complex $W.$
\end{enumerate}
\label{ThmPG}
\end{theorem}
\begin{proof}
\leavevmode
\begin{enumerate}
\item For $\Br_{\KR}$ and $\Br_{\KR \times \KZ}$ this follows from 
Lemma (\ref{LemBrRRZ}) above and functoriality.
For the groups in the filtration on $\Br_{\KR \times \KZ}$ this follows
from Lemma (\ref{LemBrFilt}) above. For 
$H^{\ast},H_M^{\ast}, H^{\ast} \oplus H^{\ast}$ this is obvious.
\item Suppose we pick another bundle $E^{'}_p$ for each
$p \in H^2(W,\KZ).$ Since $E_p,E^{'}_p$ have the same characteristic class,
they are all isomorphic. Let $\phi_p:E^{'}_p \to E_p$ be choices of isomorphisms
one for each $p \in H^2(W,\KZ).$ Each $\phi_p$ induces isomorphisms 
$\Aut(E_p) \to \Aut(E^{'}_p).$ Let $\lambda$ be any element of $\Aut(E_p)$ 
and $\lambda^{'}$ the corresponding element of $\Aut(E^{'}_p).$
Then, we have that the following diagram commutes
\begin{equation}
\CD
E_p @>{\phi}>> E^{'}_p\\
@V{\lambda}VV  @V{\lambda^{'}}VV \\
E_p @>{\phi}>> E^{'}_p.
\endCD
\nonumber
\end{equation}
Applying $G$ to each element of the above diagram shows that the sets
$G(E_p)/\Aut(E_p)$ are in bijection. Changing the isomorphisms $\phi_p$ 
doesn't affect the result.
\item It is clear that $P_G(W)$ is always a set.
All the $G_i$ described above are contravariant abelian group 
valued functors. If $f:V \to W$ is a continuous map, 
pulling back $E_p \to W$ along $f$ induces a bundle $f^{\ast}(E_p) \to V.$ 
Hence, there is a well-defined natural map $P_G(f):P_G(W) \to P_G(V)$ 
which is induced by pullback. This map doesn't depend on the choices of 
the $E_p$ as in the previous part.

If we consider $id:W \to W,$ the pullback of $E_p$ is naturally isomorphic
to $E_p$ and so the induced natural map $P_G(id):P_G(W) \to P_G(W)$ 
is the identity. This is independent of the choice of $E_p$ for the same
reason as in the previous part.

It is clear that if $f:U \to V$ and $g:V \to W,$ we have that
$f^{\ast}(g^{\ast}(E_p)) \simeq (f \circ g)^{\ast}(E_p).$ Hence,
the induced map $P_G(g \circ f)$ satisfies
$P_G(g \circ f) = P_G(f) \circ P_G(g)$ independently of the choice of $E_p.$ 
\item  Suppose we had a natural transformation $N:G_1 \to G_2.$ 
We have, for every continuous map $f:V \to W,$
and every commutative square
\begin{equation}
\CD
E^{'}_p @>F>> E_p\\
@VVV  @VVV \\
V @>f>> W
\endCD
\nonumber
\end{equation}

maps $N(E_p), N(E^{'}_p)$ such that the following
diagram commutes
\begin{equation}
\CD
G_1(E_p) @>N(E_p)>> G_2(E_p)\\
@VG_1(F)VV  @VG_2(F)VV \\
G_1(E^{'}_p) @>N(E^{'}_p)>> G_2(E^{'}_p).
\endCD
\nonumber
\end{equation}

Let $Y^1_p = G_1(E_p)/\Aut(E_p)$ and $Y^2_p = G_2(E_p)/\Aut(E_p).$
Then, we have that $N(E_p)$ induces maps $Y^1_p \to Y^2_p.$
This, in turn, induces maps $P_N(W):P_{G_1}(W) \to P_{G_2}(W).$

Similarly, from the commutative diagram in Part (2) above, 
elements of $\Aut(E_p)$ give rise to elements of $\Aut(E^{'}_p)$
by composition. Hence, the above diagram remains commutative when
we quotient each $E_p$ by $\Aut(E_p).$

Hence, for every $W$ we have that the 
following diagram commutes for every $f:V \to W$
\begin{equation}
\CD
P_{G_1}(W) @>P_N(W)>> P_{G_2}(W)\\
@V{P_{G_1}(f)}VV  @V{P_{G_2}(f)}VV \\
P_{G_1}(V) @>P_N(W)>> P_{G_2}(V).
\endCD
\nonumber
\end{equation}
Hence, we have a natural transformation $P_N:P_{G_1} \to P_{G_2}.$ 

\item Suppose $N(E_p):G_1(E_p) \to G_2(E_p)$ was always surjective for 
every principal bundle $E_p \to W,$ for every $CW$-complex $W.$ 
Then, the induced map $Y^1_p \to Y^2_p$ is always surjective (since 
quotienting both sides of $N(E_p):G_1(E_p) \to G_2(E_p)$ by $\Aut(E_p)$ 
will give a surjective map). Hence, the induced map 
$P_N(W):P_{G_1}(W) \to P_{G_2}(W)$ would be surjective as well.

The proof in the case $N(E_p)$ is injective or bijective is obvious.
\end{enumerate}
\end{proof}

Using the previous Theorem, let $P,P_2,P_3$ be the functors associated
to $\Br_{\KR},H^2,H^3$ respectively. Similarly, let $P_{3,2}, \R$ be
the functors associated to $H^2 \oplus H^3, \Br_{\KR \times \KZ}$
respectively.

\begin{lemma}
The functor $P_3$ above is the functor of Ref.\ \cite{Bunke}.
\end{lemma}
\begin{proof}
This is obvious: Both functors take the same value on $CW$-complexes
and, for any $f:V \to X,$ both act on objects via
pullback of pairs as defined in Ref.\ \cite{Bunke}.
\end{proof}

We will study the functors $P_2,P_3,P_{3,2}$ in more detail in 
Sec.\ (\ref{SecClass}).
\begin{corollary}
Suppose $W,E_p$ were as in Thm.\ (\ref{ThmPG}).
\leavevmode
\begin{enumerate}
\item The natural isomorphism $F:\Br_{\KR}(E_p) \to H^3(E_p,\KZ)$
induces a natural transformation $P \to P_3$.
\item The forgetful map (see Def.\ (\ref{DefSetup}) above) 
$F:\Br_{\KR \times \KZ}(E_p) \to \Br_{\KR}(E_p)$
induces a natural transformation $F:\R \to P_3$.
\item There is a natural surjective map $\Br_{\KR \times \KZ}(E_p) \to 
H^2(E_p,\KZ) \oplus H^3(E_p,\KZ)$. This induces a natural transformation
$\pi:\R \to P_{3,2}$. For every $CW$-complex $W,$ the induced map
$\R(W) \to P_{3,2}(W)$ is surjective.
\end{enumerate}
\label{CorPRMaps}
\end{corollary}
\begin{proof}
\leavevmode
\begin{enumerate}
\item By Ref.\ \cite{MRCMP}, Sec.\ (4.1), the forgetful map induces a 
natural isomorphism $F:\Br_{\KR}(E_p) \to H^3(E_p,\KZ)$. By the previous
Theorem, this induces a natural transformation $P \to P_3$.
\item This follows from Thm.\ (\ref{ThmPG}) items (3,4) applied to the
forgetful map (see Def.\ (\ref{DefSetup}) above) 
$F:\Br_{\KR \times \KZ}(E_p) \to \Br_{\KR}(E_p)$.
\item By Thm.\ (\ref{ThmBrauer}) item (1), we have a split short exact 
sequence
$$0 \to \ker(F) \to \Br_{\KR \times \KZ}(E_p) \underset{F}{\to} \Br(E_p)
\to 0.$$ Thus, we have a natural isomorphism 
$$\Br_{\KR \times \KZ}(E_p) \simeq H^3(E_p,\KZ) \oplus \ker(F)$$
where we have used the fact that $\Br(E_p) \simeq H^3(E_p,\KZ).$
By Thm.\ (\ref{ThmBrauer}) item (2), there is a natural surjective map
$\eta:\ker(F) \to H^2(E_p,\KZ).$ 

Hence, we have a natural map $\Br_{\KR \times \KZ}(E_p) \to H^3(E_p,\KZ) \oplus
H^2(E_p,\KZ).$ Since $F,\eta$ are both surjective, the above map is
surjective. By Thm.\ (\ref{ThmPG}) items (3,4) applied to the above
map, this induces a natural transformation $\pi:\R \to P_{3,2}.$
By the same Theorem, item (4), for every $CW$-complex $W,$ 
the functor $\pi$ induces a surjective map $\R(W) \to P_{3,2}(W).$
\end{enumerate}
\end{proof}

For a given string background $p:E_p \to W$, $P_3(W)$ encodes the
data important for T-duality, namely the $H$-flux on $E_p$ and the
characteristic class of $E_p.$ 

As we noted at the beginning of this section, for a given string background 
if we fix a closed three-form $H$ such that
$[H]$ is the $H$-flux and a two-form field $B$ such that
$H=dB$ then a large gauge transformation of the $H$-flux will change $B$
to $B'$ such that $H=dB'$. By definition, $d(B-B')=0$, i.e., $(B-B')$ is
a cohomology class in $H^2_{\text{de Rham}}(X)$. In the discussion
at the beginning of this section, it has been pointed out that for a 
large gauge transformation, this class is actually integral. 

To the data encoded in $P_3(W)$, we may add, in addition, the 
characteristic class of a large gauge transformation of the $H$-flux. 
For any $CW$-complex $W,$ this is encoded in $P_{3,2}(W):$ Here to $W$ 
we associate the triple $([p],b,H).$ The class $b$ parametrizes large 
gauge transformations of the $H$-flux and should not be 
identified with the physical Kalb-Ramond field unless $H=0.$ 

With this caveat in mind we will refer to the class $b$ as `the $B$-field'
or `the $B$-class' in what follows. It should be clear, however, 
that this should not, in fact, be interpreted as 
the physical Kalb-Ramond field, 
but should be viewed as a shift in the Kalb-Ramond field $B$ 
that is, as the characteristic
class of a large gauge transformation of the $H$-flux.

We now point out some relations among the functors 
$P_2,P_3,P_{3,2}$ described above. We point out the action of
the transformations induced by Topological T-duality on these functors. 
We show how these indicate a way to
answer the question posed in Sec.\ (1) in the paragraph 
after Lemma (\ref{LemATMP}) (see Eq.\ (\ref{EqDefT32}) below).
We close with a comparision to the work of Schneider 
(see Ref.\ \cite{Schneider}). 

Let $W$ be as above. Consider the functor 
$P_2$ above. As pointed out above after 
Lemma (\ref{LemTR}), for $W,E_p$ as above, by considering 
$\Br_{\KR \times \KZ}(E_p)/B_1,$ we are led to consider `triples' 
of the form {\bf (principal bundle $E_p$, Class in $H^2(E_p,\KZ)$,
Class in $H^3(E_p,\KZ)$)}. The isomorphism classes of such triples
over $W$ are $P_{3,2}(W)$ where
$P_{3,2}$ is the functor defined before Cor.\ (\ref{CorPRMaps}) above. 
Also, by the discussion in Cor.\ (\ref{CorPRMaps}) item (3) above, 
there is a natural transformation $\pi:\R \to P_{3,2}.$

In Sec.\ (\ref{SecTriple}) we prove that there is a well-defined
map $T_{3,2}:R_{3,2} \to R_{3,2}$ inducing a natural transformation
$T_{3,2}(W):P_{3,2}(W) \to P_{3,2}(W)$ by an argument involving the 
classifying space of triples $R_{3,2}$. We also show that this map
induces the following commutative diagram.
\begin{equation}
\CD
T_{3,2}(W) @>{T_{3,2}(W)}>> T_{3,2}(W)\\
@V{\pi(W)}VV  @V{\pi(W)}VV \\
T_3(W) @>{T(W)}>> T_3(W)
\endCD
\label{EqDefT3}
\end{equation}
i.e., $\pi \circ T_{3,2} = T \circ \pi$ as natural transformations.

Given a $C^{\ast}$-dynamical system 
$(\A, \KR \times \KZ, \alpha)$ corresponding to a class
$a \in \Br_{\KR \times \KZ}(E_p),$ adding an element of $B_1$ to
$a$ will not change the Phillips-Raeburn invariant of $\alpha|_{\KZ}$
or the $H$-flux $F(a).$ We showed in
Ref.\ \cite{ATMP} that under T-duality the Phillips-Raeburn invariant of the
$\KZ$-action on $\CrPr{\A}{\KR}{\alpha|_{\KR}}$ associated to $T(a)$ 
only depends on the $H$-flux and the Phillips-Raeburn invariant
of the dynamical system associated to $a.$ It doesn't depend on the
lift of these data to a $\KR \times \KZ$-action on $\A.$ 
Hence, the following diagram commutes
\begin{equation}
\CD
\R(W) @>{T_R}(W)>> \R(W)\\
@V{\pi(W)}VV  @V{\pi(W)}VV \\
P_{3,2}(W) @>{T_{3,2}(W)}>> P_{3,2}(W).
\endCD
\label{EqDefT32}
\end{equation}

Also, by Cor.\ (\ref{CorPRMaps}), item(3), the map 
$\pi(W):\R(W) \to P_{3,2}(W)$ 
induced by the functor $\pi$ in Cor.\ (\ref{CorPRMaps}) item (3), is 
always surjective. Thus we may infer properties 
of $T_R$ from those of $T_{3,2}$ but it
should be clear that they are not the same.
In this paper we mainly study $P_{3,2}$ and $T_{3,2}.$ This is sufficient
to answer the question we raised in Section (1) in the paragraph 
after Lemma (\ref{LemATMP}).

Let $W$ be as above. Consider the functor $P_3$ above.
In Ref.\ \cite{Bunke},
the authors show that there is a classifying space $R_3$ for the functor
$P_3$ and $P_3(W) = [W,R_3].$ Further, Ref.\ \cite{Bunke} shows
that that T-duality is a natural transformation from $P_3$ to itself
giving a map we denote as $T_3(W): P_3(W) \to P_3(W)$. 

By Cor.\ (\ref{CorPRMaps}) above, there is a natural transformation 
$P \to P_3$ induced by the isomorphism $\Br_{\KR}(E_p) \to H^3(E_p,\KZ)$
which by Part (5) of Thm.\ (\ref{ThmPG}) induces a natural isomorphism
$P(W) \to P_3(W)$ for every $CW$-complex $W.$

Hence, for every $CW$-complex $W$, we have an isomorphism of sets
\begin{gather}
\amalg_{[p]\in H^2(W,\KZ)} H^3(E_p,\KZ)/\Aut(E_p) \simeq
\amalg_{[p] \in H^2(W,\KZ)} \Br_{\KR}(E_p)/\Aut(E_p).\nonumber
\end{gather}
Therefore, an isomorphism class of a pair over $W$ 
(in the sense of Ref.\ \cite{Bunke}) 
determines and is determined by an element of $P.$ 
In particular, for any CW complex $W,$ we have that $[W,R_3] = P(W)$ 
as well. Thus, the Topological T-duality functor of Bunke et al may
be easily derived from the $C^{\ast}$-algebraic theory of Ref.\ \cite{MRCMP}
for circle bundles (see also Ref.\ \cite{Schneider} Prop.\ (2.8)).

Let $\R$ be as defined above. By the above, this is a set-valued functor on the
category of unbased $CW$-complexes. Then, we have the following theorem
\begin{lemma}
There is a well-defined map $T_R: \R(W) \to \R(W)$ 
induced by the crossed product.
\label{LemTR}
\end{lemma}
\begin{proof}
We need the following well-known fact.
Let $\A, \B$ be $C^{\ast}$-algebras with $G$-action $\alpha,\beta$ 
respectively. Let $C_c(G,\A)$ the $\alpha$-twisted convolution algebra
of $\A$-valued functions on $G$ which are of compact support on $G$.
Similarly, let $C_c(G,\B)$ be the $\beta$-twisted convolution algebra
of $\B$-valued functions on $G$ which are of compact support on $G$.
Give $C_c(G,\A)$ and $C_c(G,\B)$ the inductive limit topology\footnote{ 
See Ref.\ \cite{Will}, Corollary 2.48.}.
\begin{theorem}
Suppose that $(\A,G,\alpha)$ and $(\B,G,\beta)$ are dynamical systems
and $\phi:\A \to \B$ is an equivariant homomorphism. Then, there is a
homomorphism 
$\phi \rtimes \operatorname{id}: \CrPr{\A}{G}{\alpha} \to \CrPr{\B}{G}{\beta}$
mapping $C_c(G,\A)$ into $C_c(G,\B)$ such that 
$\phi \rtimes \operatorname{id}(f)(s) = \phi(f(s))$.
\end{theorem}
From the proof of this theorem, it is clear that the extension
$\phi \rtimes \operatorname{id}$ is unique.

Here, $\A = \B$ and $G = \KR$. 
Suppose $y \in \Br_{\KR \times \KZ}(W),$ and we pick a 
representative $x=(\A, \alpha \times \phi)$ of $y$. 
We may define $T(y)$ to be the dynamical system
$(\CrPr{\A}{\KR}{\alpha}, \alpha^{\#} \times \phi^{\#})$
where $\phi^{\#}$ is the map induced on the crossed product
by $\phi^{\#}(f)(t) = \phi \rtimes \operatorname{id}(f)(t)$.
It is clear that it is unique and commutes with the
$\KR$-action.

Changing the representative to a Morita equivalent one 
$y=(\A', \alpha' \times \phi')$ will not change the
Morita equivalence class of the answer because, by 
Lemma (3.1) of Ref.\ \cite{CKRW}, $\alpha' \times \phi'$ is
outer conjugate to $\alpha \times \phi.$ 
Hence, there is an isomorphism $\Phi:\A \to \A$ such that
$\Phi \circ (\alpha \times \phi) \Phi^{-1}$ is exterior 
equivalent to $(\alpha' \times \phi').$ Hence, there is a
continuous map $u:\KR \times \KZ \to UM(\A)$ such that
$u_{gh} = u_g (\alpha \times \phi)_g(u_h)$ such that 
$(\alpha' \times \phi')_g = \Ad u_g \circ (\alpha \times \phi)_g.$
Then the suspension $u'(t) = u, \forall t \in \KR$ gives an
exterior equivalence between $\alpha^{\#} \times \phi^{\#}$ and
$\alpha'^{\#} \times \phi'^{\#}.$


\end{proof}

We now remark on some properties of the functor $\R$ and their
connection with Topological T-duality. 
We begin by considering the functor $P_{3,2}$ described above.
Apart from the construction above, there is another reason to look 
at $P_{3,2}(W):$ By Thm.\ (\ref{ThmBrauer}) above,  $\Br_{\KR \times \KZ}(E_p)$ 
has a large continuous part, namely the elements of 
$H^2_M(\KR \times \KZ,C(E_p,\KT)).$ If we quotient 
$\Br_{\KR \times \KZ}(E_p)$ by this continuous part, we obtain 
the group $H^2(E_p,\KZ) \oplus H^3(E_p,\KZ)$ associated to $E_p$. 
Alternatively, we may consider $\im(F_1)$ where $F_1:\Br_{\KR \times \KZ}(E_p)
\to \Br_{\KZ}(E_p)$ is the forgetful map of Def.\ (\ref{DefSetup}).
In either case, for any $CW$-complex $W,$ we have to consider the functor
$P_{3,2}(W)=${\bf$($principal bundle, closed integral two-form, $H$-flux $)$ }
over $W.$ We have the following property of $P_{3,2}:$

\begin{lemma}
Let $T,T_3$ be as defined above.
The map $F: \Br_{\KR \times \KZ}(X) \to \Br_{\KR}(X)$ in Def.\ (\ref{DefSetup})
above gives a natural transformation $F$ between the functors $\R$ and $P_3$
defined above. We have that $F \circ T = T_3 \circ F.$
\end{lemma}
\begin{proof}
By Part (2) of Cor.\ (\ref{CorPRMaps}) above, the map 
$F:\Br_{\KR \times \KZ}(X) \to \Br_{\KR}(X)$
gives a natural transformation between the functors $\R$
and $P_3.$ By construction, we know that if we `forget' the $\KZ$-action 
on $\A$, the crossed product by the $\KR$-action doesn't change in 
either $H$-flux or topology. That is, if we have two T-dual (in the sense above)
$C^{\ast}$-dynamical systems $(\A,\KR \times \KZ, \alpha)$ and
$(\B, \hat{\KR} \times \KZ, \beta)$ then the dynamical systems
$(\A, \KR, \alpha|_{\KR})$ and $(\B, \hat{\KR}, \beta|_{\hat{\KR}})$ are
T-dual in the sense of Ref.\ \cite{MRCMP}. 
Hence, the following diagram commutes

\begin{equation}
\CD
\R(W) @>T(W)>> \R(W)\\
@VF(W)VV  @VF(W)VV \\
P_3(W) @>{{T_3}(W)}>> P_3(W).
\endCD
\nonumber
\end{equation}
This implies that $F \circ T = T_3 \circ F.$ 
\end{proof}

Now, by the discussion in Cor.\ (\ref{CorPRMaps}) item (3) above, 
there is a natural transformation $\pi:\R \to P_{3,2}.$

In Sec.\ (\ref{SecTriple}) we prove that there is a classifying space
for $P_{3,2}$ denoted $R_{3,2}$ such that $P_{3,2}(W) = [W,R_{3,2}].$
We show that there is a well-defined map $T_{3,2}:R_{3,2} \to R_{3,2}$ 
inducing a natural transformation $T_{3,2}(W):P_{3,2}(W) \to P_{3,2}(W)$. 
We also show that this map induces the following commutative diagram.
\begin{equation}
\CD
T_{3,2}(W) @>{T_{3,2}(W)}>> T_{3,2}(W)\\
@V{\pi(W)}VV  @V{\pi(W)}VV \\
T_3(W) @>{T(W)}>> T_3(W)
\endCD
\label{EqDefT3}
\end{equation}
i.e., $\pi \circ T_{3,2} = T \circ \pi$ as natural transformations.

Given a $C^{\ast}$-dynamical system 
$(\A, \KR \times \KZ, \alpha)$ corresponding to a class
$a \in \Br_{\KR \times \KZ}(E_p),$ adding an element of $B_1$ to
$a$ will not change the Phillips-Raeburn invariant of $\alpha|_{\KZ}$
or the $H$-flux $F(a).$ We showed in
Ref.\ \cite{ATMP} that under T-duality the Phillips-Raeburn invariant of the
$\KZ$-action on $\CrPr{\A}{\KR}{\alpha|_{\KR}}$ associated to $T(a)$ 
only depends on the $H$-flux and the Phillips-Raeburn invariant
of the dynamical system associated to $a.$ It doesn't depend on the
lift of these data to a $\KR \times \KZ$-action on $\A.$ 
Hence, the following diagram commutes
\begin{equation}
\CD
\R(W) @>{T_R(W)}>> \R(W)\\
@V{\pi(W)}VV  @V{\pi(W)}VV \\
P_{3,2}(W) @>{T_{3,2}(W)}>> P_{3,2}(W).
\endCD
\label{EqDefT32}
\end{equation}

Also, by Cor.\ (\ref{CorPRMaps}), item(3), the map 
$\pi(W):\R(W) \to P_{3,2}(W)$ 
induced by the functor $\pi$ in Cor.\ (\ref{CorPRMaps}) item (3), is 
always surjective. Thus we may infer properties 
of $T_R$ from those of $T_{3,2}$ but it
should be clear that they are not the same.
In this paper we mainly study $P_{3,2}$ and $T_{3,2}.$ This is sufficient
to answer the question we raised in Section (1) in the paragraph 
after Lemma (\ref{LemATMP}).

In a recent thesis by Schneider \cite{Schneider} the author defines
the collection of equivalence classes of $C^{\ast}$-dynamical systems 
$(\A, G, \alpha)$ whose spectrum is a principal $G/N$- bundle 
over $W$ (denoted $\text{Dyn}^{+}(W)$).
He then defines a T-duality map induced by the crossed product between
equivalence classes of such dynamical systems. The resulting dynamical
system is $(\CrPr{\A}{\hat{G}}{\alpha^{\#}},\hat{G},\alpha^{\#})$ and
has spectrum a principal $\hat{G}/N^{\perp}$-bundle over $W.$ 

Thus the duality map in Ref.\ \cite{Schneider} would map 
systems of the form $(\A, \KR \times \KZ, \alpha)$ with spectrum
a principal $S^1$-bundle over $W$ to
those of the form $(\A^{\#}, \hat{\KR} \times \KT, \hat{\alpha})$
with a spectrum-fixing $\KT$-action and spectrum a principal
$S^1$-bundle over $W.$

This map is not the same as the
T-duality map we are considering here, since
we map equivalence classes of $C^{\ast}$-dynamical
systems of the form $(\A, \KR \times \KZ, \alpha)$ with spectrum
a principal $S^1$-bundle over $W$ and a spectrum-fixing $\KZ$-action
on $\A$ to
those of the form $(\A^{\#}, \hat{\KR} \times \KZ, \hat{\alpha}).$
Here $\A^{\#}$ has spectrum a principal $S^1$-bundle over $W$ but
there is a {\em spectrum-fixing} $\KZ$-action on $\A^{\#}.$

\section{The Classifying Space of $k$-pairs\label{SecClass}}
In this section, we study the functors $P_2,P_{3,2}$ defined above.
We show that they are representable and possess classifying spaces
$R_2,R_{3,2}$ respectively. We calculate some properties of these
classifying spaces.

Let $\CSet$ be the category of sets with functions as morphisms.
Let $\mathcal{C}$ be the category of unbased CW complexes with
unbased homotopy classes of continuous maps as morphisms.
Let $\mathcal{C}_0$ be the category of CW complexes which
are finite subcomplexes of some fixed countably infinite dimensional
standard simplicial complex.
$(\mathcal{C},\mathcal{C}_0)$ is a homotopy category in the sense of 
Ref.\ \cite{Brown} (see Thm.\ (2.5) in \cite{Brown}).

Let $W$ be a {\em fixed} CW complex. We define a $k$-pair over $W$
to consist of a principal $S^1$-bundle $p:E \to W$ together
with a cohomology class $b \in H^k(W,\KZ)$.
We denote a $k$-pair as $([p],b)$. (Here the space $W$ is understood
from the context as is the value of $k$.)
Note that a `pair' in the sense of Ref.\ \cite{Bunke} would be a 
termed a $3$-pair here.

\begin{definition}
We declare two $k$-pairs ({\bf same} $k$) $([p],b)$ and $([q],b')$  over $W$ 
equivalent if 
\begin{itemize}
\item We are given two principal $S^1$-bundles
$p:E \to W$ and $q:E' \to W$ such that 
\begin{equation}
\CD
E @>\phi>> E'\\
@VpVV  @VqVV \\
W @>id>> W
\endCD
\nonumber
\end{equation}
commutes.

\item We also require that $b' = \phi^{\ast}(b).$
\end{itemize}
\end{definition}

It is clear that the collection of equivalence classes of $k$-pairs
over a fixed space $W$ (denoted $P_k(W)$) is a {\em set}. For all $W$,
we have a distinguished pair consisting of the trivial $S^1$-bundle
over $W$ with the zero class in $H^k(W,\KZ)$. Thus, $P_k(W)$ is actually
a pointed set.

\begin{definition}
Let $W$ and $Y$ be two CW-complexes and let $f:W \to Y$ be a continuous map.
Let $([p],b) \in P_k(Y)$ be represented by
a principal $S^1$-bundle $p:E \to Y$ and a class $b \in H^k(Y,\KZ)$.
We define the pullback of $([p],b)$ via $f,$
denoted $f^{\ast}([p],b),$ to be the following data
\begin{itemize}
\item The unique principal $S^1$-bundle $f^{\ast}p: f^{\ast}E \to W$ such that
the following diagram commutes
\begin{equation}
\CD
f^{\ast}E @>\tilde{\phi_f}>> E\\
@Vf^{\ast}pVV  @VpVV \\
W @>f>> Y.
\endCD
\nonumber
\end{equation}
\item The cohomology class $\phi_f^{\ast}(b)$ in $H^k(f^{\ast}E,\KZ)$.
\end{itemize}
That is,we define $f^{\ast}([p],b) = (f^{\ast}[p], \phi_f^{\ast}(b))$.
\end{definition}

\begin{lemma}
Let $f_0,f_1:W \to Y$ be freely homotopic. For any pair
$([p],b) \in P_2(Y),$ $f_0^{\ast}([p],b)$ is equivalent
to $f_1^{\ast}([p],b)$.
\end{lemma}
\begin{proof}
Let $p:E \to Y$ be a principal $S^1$-bundle. 
We have pullback squares for $i=0,1$
\begin{equation}
\CD
f_i^{\ast}E @>\tilde{\phi_i}>> E\\
@Vf_i^{\ast}pVV  @VpVV \\
W @>f_i>> Y.
\endCD
\nonumber
\end{equation}

Then, by Ref.\ \cite{VBKT},
Cor.\ (1.8), the pullback bundles $f_0^{\ast}p:f_0^{\ast}E \to W$ and
$f_1^{\ast}:f_1^{\ast}E \to W$ are isomorphic. Further, by the same
lemma, this isomorphism is implemented by a map 
$\psi:f_0^{\ast}E \to f_1^{\ast}E$. This map induces isomorphisms
on the cohomology groups such that 
$ \psi \circ f_0^{\ast} = f_1^{\ast}. $
As a result, by the above definition, 
$f_0^{\ast}([p],b) = f_1^{\ast}([p],b)$.
\end{proof}

Hence, $P_k(W)$ is a pointed set depending only on the homotopy type
of $W$. Given a map $f:W \to Y,$ define $P_k(f): P_k(Y) \to P_k(W)$ to
be the map induced by pullback of pairs. 
It is clear that $P_k(1) = \operatorname{Id}$. 
(This is just the condition that two pairs be equivalent).
Hence, $P_k$ extends
to a functor (also denoted $P_k$) $P_k: \mathcal{C} \to \CSet$.

\begin{theorem}
For every $k,$ the functor $P_k$ above satisfies the conditions
of the Brown Representability 
Theorem.  Hence, for every $k$, there exists a classifying space $R_k$ for 
$P_k$.
\end{theorem}
\begin{proof}
There are two conditions we need to prove.
\begin{enumerate}
\item Consider an arbitrary family $\{W_{\mu}\}, \mu \in I$ of 
of objects in $\mathcal{C}$. 
Let $Y = \bigsqcup_{\mu \in I} W_{\mu}$.
Let $h_{\mu}: W_{\mu} \to \bigsqcup_{\mu \in I} W_{\mu}$ 
be the inclusion maps.

We have a pullback square (for every $\mu \in I$)
\begin{equation}
\CD
h_{\mu}^{\ast}E @>\tilde{h}_{\mu}>> E\\
@Vp_{\mu}VV  @VpVV \\
W_{\mu} @>h_{\mu}>> Y.
\endCD
\nonumber
\end{equation}
Here $p_{\mu} = h_{\mu}^{\ast}p$.
Since $H^2(W,\KZ) \simeq \prod_{\mu \in I} H^2(W_{\mu},\KZ)$,
we have that $[p] = ([p_{\mu}]), \mu \in I$.

Let $h_{\mu}^{\ast}E = E_{\mu},$ then,
we also have that $E = \bigsqcup_{\mu \in I} E_{\mu}$ and
$H^k(E,\KZ) \simeq \prod_{\mu \in I} H^k(E_{\mu},\KZ)$.
Hence, every class $b \in H^k(E,\KZ)$ may be written as
$(b_{\mu}), \mu \in I$ with $b_{\mu} = \tilde{h}_{\mu}^{\ast}(b)$.
Hence, we have an isomorphism
$$ \Pi_{\mu} P(h_{\mu}): P(\bigsqcup_{\mu} W_{\mu}) \approx 
\Pi_{\mu} P(W_{\mu}). $$

\item Suppose we are given CW complexes $A,W_1,W_2$ and
continuous maps $f_i: A \to W_i, g_i: W_i \to Z, 
i=1,2$ such that
\begin{equation}
\CD
A @>f_1>> W_1\\
@Vf_2VV  @Vg_1VV \\
W_2 @>g_2>> Z
\endCD
\nonumber
\end{equation}
commutes up to homotopy and is a pushout square in $\mathcal{C}.$ 
We may take $f_i$ to be inclusions into the $W_i$ and
$Z$ the result of gluing $W_1$ to $W_2$ along $A$.
Suppose $u_i \in P(W_i)$ satisfy $P(f_1)u_1 = P(f_2)u_2$.
For $i=1,2,$ let $u_i$ correspond to the pair $([p_i],b_i)$ over $W_i,$
where $p_i:E_i \to W_i$ are principal $S^1$-bundles.
Then, since $P(f_1)u_1 = P(f_2)u_2,$ 
$f_1^{\ast}E_i \simeq f_2^{\ast}E_2$. This implies that
the restrictions of $f_i^{\ast}E_i$ to $A$ are the same.
Hence, these two bundles may be glued into a unique 
bundle $p:E \to Z$. Note that $G_i = f_i^{\ast}E_i \subset E, i = 1,2$
and $G_1 \cup G_2 = E$.
We have a pullback square
\begin{equation}
\CD
G_i @>\tilde{g}_i>> E\\
@Vp_iVV  @VpVV \\
W_i @>g_i>> Z
\endCD
\nonumber
\end{equation}
 
By the Mayer-Vietoris theorem, we have
$$
H^k(E,\KZ) \to H^k(G_1,\KZ) \oplus H^k(G_2,\KZ) \to H^k(G_1 \cap G_2,\KZ)
$$
Now $f_1^{\ast}(b_1) = f_2^{\ast}(b_2)$ and so the image of
$(b_1,b_2)$ via the second map above is zero. Hence, by exactness,
there is an element in $c \in H^k(E,\KZ)$ such that 
$\tilde{g}_i(c) = b_i, i =1,2$.
Thus, we define an element $v \in P(Z)$ 
by $v = ([p],c)$. It is clear that $P(g_i)v = u_i, i = 1,2$.
\end{enumerate}

As a result, for every $k$, there is a CW complex $R_k$ such that
isomorphism classes of $k$-pairs over a space
$W$ correspond to {\em unbased} homotopy classes of maps
from $W \to R_k$.     
\end{proof}

Similarly we may define a $k_1,k_2,\ldots,k_n$-tuple ($k_i \in \KN)$
over a space $W$ to consist\footnote{Obviously the ordering of the
$k_i$ is irrelevant.}
of a principal $S^1$-bundle $p:E \to W$ together
with cohomology classes $b_1, \ldots, b_n$ such that $b_i \in H^{k_i}(E,\KZ).$ 
Exactly as above we may define the notion of equivalent tuples and
a set valued functor $P_{k_1,k_2,\ldots,k_n}(W).$
As above, such a functor is representable and has a representation
space denoted $R_{k_1,k_2,\ldots,k_n}.$ Note that there are 
natural maps $R_{k_1,\ldots,k_n} \to K(\KZ,2)$ and
$R_{k_1,\ldots,k_n} \to K(\KZ,k_i-1)$ given by sending
$([p],b_1,\ldots, b_i) \to [p]$ and 
$([p],b_1,\ldots,b_i) \to p_!(b_i)$ respectively.
In addition, the canonical $S^1$-bundle $U$ over $K(\KZ,2)$ defines 
a unique pair $(U,0)$ over $R_k$ for every $k$ and hence the natural map
$R_k \to K(\KZ,2)$ is naturally split. This implies that 
the cohomology ring of $R_k$ contains $\KZ[\alpha], \alpha \in H^2(R_k,\KZ).$

Given a $3,2$-tuple $([p],b,H)$ over $W,$ we obtain a unique 
$2$-pair $([p],b)$ and $3$-pair $([p],H)$ over $W.$ Similarly, given
a $2$-pair $([p],b)$ and a $3$-pair $([p],H)$ ({\em same } $[p]$)
over $W,$ we obtain a unique triple $([p],b,H)$ over $W.$ 
Thus, $R_{3,2}$ is a fiber product $R_3 \underset{K(\KZ,2)}{\times} R_2.$ 
Similarly 
$R_{k_1,\ldots,k_n} \simeq R_{k_1} \underset{K(\KZ,2)}{\times}
\ldots \underset{K(\KZ,2)}{\times} R_{k_n}.$
There are natural fibrations $q_{k_i}:R_{k_1,\ldots,k_n} \to R_{k_i}.$

For the sake of completeness we note the following 
\begin{lemma}
Suppose $W$ is $k$-connected, $k \geq 2.$ Then,
$P_{k+1}(W) \simeq H^{k+1}(W,\KZ).$
\end{lemma}
\begin{proof}
Since $W$ is at least $2$-connected, all principal $S^1$-bundles
on $W$ are trivial. Further, $H^{k+1}(W \times S^1,\KZ) 
\simeq H^{k+1}(W,\KZ)$ and hence the result.
\end{proof}

We will only consider $2$-tuples, $3$-tuples and $3,2$-tuples in this
paper. We will also consider $P_2,P_3,P_{3,2}$ and the corresponding
classifying spaces $R_2, R_3, R_{3,2}.$

Bunke et al.\ \cite{Bunke} have considered the case $k=3$.
We denote their classifying space $R_3$ here.
For the remainder of this section and the next we work with $k=2$. We
abbreviate $2$-pair to `pair'.

A priori, $R_2$ is an unbased CW complex. We now arbitrarily pick a
basepoint $r_0$ in $R_2$. 

\begin{lemma}
Let $W$ be {\em any} CW complex. Pick a basepoint $x_0 \in W.$
Any unbased map $f:W \to R_2$ may be freely homotoped to a
based map $g:(W,x_0) \to (R_2,r_0).$
\label{R2BpLemma}
\end{lemma}
\begin{proof}
Suppose $W$ was {\em any} CW complex,
and $f:W \to R_2$ an unbased map. 
By the Lemma that follows, we know that $R_2$ is a 
fibration of a connected space over
a connected base space. Hence $R_2$ is connected. 
Pick a basepoint $x_0 \in W$. Pick a path $q: I \to R_2$ connecting
$f(x_0)$ to $r_0$. Extend the data $f,q$ to a free homotopy
$H:W \times I \to R_2$. Then, $g = H(1,.):W \to R_2$ is 
map such that $g(x_0) = r_0$. The map $g$ classifies the same pair that
$f$ does, since $R_2$ is an {\em unbased classifying space}.
\end{proof}

\begin{lemma}
There is a fibration $K(\KZ,2) \to R_2 \to K(\KZ,2) \times K(\KZ,1)$.
\label{R2FibLemma}
\end{lemma}
\begin{proof}
Given a pair $([p],b)$ over $W$, we obtain two natural
cohomology classes $[p] \in H^2(W,\KZ)$ and $p_!(b) \in H^1(W,\KZ)$.
As a result, there is a natural map 
$\phi \times \psi: R_2 \to K(\KZ,2) \times K(\KZ,1)$.
Given a pair $([p],b)$ over any space $W$,
classified by $f: W \to R_2$, the map 
$f \mapsto \phi \circ f$ corresponds
to the map $([p],b) \mapsto [p]$. Similarly,
$f \mapsto \psi \circ f$ corresponds to the map
$([p],b) \mapsto p_!(b)$. 
We pick a basepoint in $K(\KZ,2) \times K(\KZ,1)$ such
that $\phi \times \psi$ is a based map.
Suppose $f:W \to R_2$ classified a pair $([p],b)$ over $W$.
Pick a basepoint $x_0 \in W.$ By Lemma (\ref{R2BpLemma}) above,
$f$ may be freely homotoped to a based map $g:(W,x_0) \to (R_2,r_0)$.
Suppose $g$ was in the homotopy fiber of $\phi \times \psi$.
Then, we would obtain a pair $([p],b)$ over $W$ such that 
$p_!(b) = 0, [p] = 0$. This would correspond to the
trivial bundle $W \times S^1 \to W$ equipped with the
cohomology class $1 \times a, a \in H^2(W,\KZ)$. Hence we would
get a natural based map $W \to K(\KZ,2)$. 
Conversely, given a class $a$ in $H^2(W,\KZ)$, we could obtain
a pair $(0, 1 \times a)$ over $W$ which would have $[p]=0$ and
$p_!(1 \times a) = 0$. By Lemma (\ref{R2BpLemma}), this pair would
be classified by a based map $g:(W,x_0) \to (R_2,r_0)$. Obviously,
$(\phi \times \psi)\circ g$ would be nullhomotopic.

Thus, the homotopy fiber of $\phi \times \psi$ is $K(\KZ,2)$.
\end{proof}

\begin{lemma}
The homotopy groups of $R_2$ are as follows
\begin{itemize}
\item $\pi_1(R_2) = \KZ$, 
\item $\pi_2(R_2) \simeq \KZ^2$,
\item $\pi_i(R_2) = 0, i > 2$.
\end{itemize}
\end{lemma}
\begin{proof}
We had picked a basepoint for $R_2$. Hence, we may calculate $\pi_i(R_2)$
from the long exact sequence of the fibration in Lemma (\ref{R2FibLemma}). 
We find 
that the nozero part of the sequence is
$$ 
0 \to \KZ \to \pi_2(R_2) \to \KZ \to 0 \to \pi_1(R_2) \to \KZ \to 0.
$$

Thus, $\pi_1(R_2) = \KZ, \pi_2(R_2) \simeq \KZ^2$,
and  $\pi_i(R_2) = 0, i > 2$.
\end{proof}

We may characterize $R_2$ as follows
\begin{lemma}
Let $c:K(\KZ,2) \times K(\KZ,1) \to K(\KZ,3)$ be the based map which
induces the cup product. Then $R_2$ is the homotopy fiber of $c$.
\end{lemma}
\begin{proof}

Let $f:W \to R_2$ be a map inducing the pair $([p],b)$ over $W$. 
Fixing a basepoint $x_0 \in W,$ we may replace $f$ by a based map 
$g:(W,x_0) \to (R_2,r_0)$ by Lemma (\ref{R2BpLemma}). It is 
clear that we may take $\phi \times \psi$ to be based.
Then we have a principal $S^1$-bundle $p:E \to W.$ By the Gysin
sequence of this bundle we have that $[p] \cup p_!(b) = 0$.
This implies that $c\circ(\phi \times \psi) \circ f$ is nullhomotopic
via a based homotopy,
since $c$ is exactly the based map which gives the cup product.

Conversely, suppose we are given a based map $f:W \to K(\KZ,2) \times K(\KZ,1)$
such that $c \circ f$ is nullhomotopic. Then, this corresponds to
a class $a \in H^1(W,\KZ)$ and a class $[p] \in H^2(W,\KZ)$ such
that $[p] \cup a = 0.$ Pick a principal $S^1$-bundle 
$p:E \to W$ with characteristic class $[p]$. From the Gysin sequence of this
bundle we see that $[p] \cup a = 0$ implies that $a = p_!(b)$ for
some $b \in H^2(E,\KZ)$. Thus, we obtain a pair $([p],b)$ over
$W$ and hence an unbased map $g:W \to R_2$. 
By the above argument, we may replace it with a based map $h$ classifying
the {\em same} pair over $W.$ Obviously, 
$(\phi \times \psi)\circ h = f$ as a based map.

Hence, $R_2$ is the homotopy fiber of the based map $c$ 
in the category of based CW complexes with basepoint preserving 
homotopy classes of maps between them. 
There is a forgetful functor from this category to the category
$\mathcal{C}.$ We take the image of the homotopy fiber of $c$ via this
functor. This determines $R_2$ up to homotopy equivalence in $\mathcal{C}.$
\end{proof}

Since $\pi_1(R_2) \neq 0,$ the choice of basepoints might be important.
Indeed, we have the following
\begin{lemma}
The space $R_2$ is not simple.
\end{lemma}
\begin{proof}
Suppose $R_2$ was simple: Then, from Postnikov theory, we see that
$R_2$ would be homotopy equivalent to the product
$K(\KZ,1) \times K(\KZ,2) \times K(\KZ,2)$ via a based homotopy.
Given $f:W \to R_2$ by Lemma (\ref{R2BpLemma}), 
we could obtain a based map $g:W \to R_2$ classifying the same pair
over $W$ as $f$. Hence we would obtain based maps
$W \to K(\KZ,2)$, $W \to K(\KZ,2)$ and $W \to K(\KZ,1)$.
The pair would then be {\em determined} by classes $[p],a \in H^2(W,\KZ)$
and $p_!(b) \in H^1(W,\KZ)$. Here $[p]$ would be the characteristic
class of a principal $S^1$-bundle $p:E \to W$. 
This would imply in turn that $b$ would be
determined by $p_!(b)$ and $a$ and hence that the Gysin sequence
for $p:E \to W$ would split at degree two for {\em any} principal $S^1$-bundle
$E$ over $W$. Since $W,E$ were arbitrary, this is obviously impossible.
\end{proof}

\begin{lemma}
The cohomology of $R_2$ up to degree $3$ is
\begin{itemize}
\item $H^0(R_2,\KZ) \simeq \KZ,$
\item $H^1(R_2,\KZ) \simeq \KZ,$
\item $H^2(R_2,\KZ) \simeq \KZ.$
\item $H^3(R_2,\KZ) \simeq \KZ.$
\end{itemize}
\end{lemma}
\begin{proof}
Consider the fibration
$K(\KZ,2) \to R_2 \to K(\KZ,2) \times K(\KZ,1)$. We have that 
$H^{\ast}(K(\KZ,2),\KZ) \simeq \KZ[a]$ where $a$ is a generator
of $H^2(K(\KZ,2),\KZ) \simeq \KZ.$ This ring has no automorphisms apart from
$a \to -a.$ Since the fibration is oriented, 
the generator of $\pi_1(K(\KZ,2) \times K(\KZ,1))$ 
acts trivially on the cohomology of $K(\KZ,2).$ 
As a result, we may use the Serre spectral sequence 
using cohomology with untwisted coefficients to calculate $H^{\ast}(R_2,\KZ).$

We note that the above fibration is pulled back from the
path-loop fibration over $K(\KZ,3)$ via the map 
$c:K(\KZ,2) \times K(\KZ,1) \to K(\KZ,3)$ inducing the cup product.
Suppose $\mu$ was a generator of $H^1(K(\KZ,1),\KZ)$ and 
that $\lambda$ was a generator of $H^2(K(\KZ,2),\KZ).$
Let $\tilde{\mu} = (\phi \times \psi)^{\ast}(\mu),$ and
$\tilde{\lambda} = (\phi \times \psi)^{\ast}(\lambda).$
Then we have that $\tilde{\mu} \cup \tilde{\lambda} = 0.$
This shows that the transgression 
$E_2^{2,0} \to E_2^{0,3}$ must be a map 
$k:H^2(K(\KZ,2),\KZ) \simeq \KZ \to H^2(K(\KZ,2) \times K(\KZ,1),\KZ) 
\simeq \KZ$ sending $a \to a (\mu \cup \lambda).$

From the spectral sequence table, we see that 
$H^0(R_2,\KZ) \simeq \KZ,$ $H^1(R_2,\KZ) \simeq \KZ,$ 
$H^2(R_2,\KZ) \simeq \KZ$ and $H^3(R_2,\KZ) \simeq \KZ.$
\end{proof}

Note that the canonical bundle over $K(\KZ,2)$ gives rise to a unique
pair $(1,0)$ over $K(\KZ,2).$ This is classified by a map
$K(\KZ,2) \to R_2$ which is a section of the natural map $R_2 \to K(\KZ,2).$
Hence, the cohomology ring of $R_2$ contains $\KZ[b], b \in H^2(R_2,\KZ).$

This ring has generators $a,b,c$ in degrees $1,2$ and $3$ respectively
such that 
\begin{equation}
a.b = 0, b^n \neq 0 {\mbox{ for any } n. }  \label{R2Rel}
\end{equation}

We now determine the action of $\pi_1(R_2)$ on $\pi_2(R_2).$
\begin{theorem}
The action of the generator $S$ of $\pi_1(R_2) \simeq \KZ$ 
on $\pi_2(R_2)$ is given by
$$
S(a,b) = (a + b, b).
$$
\label{ThmR2Action}
\end{theorem}
\begin{proof}
From the long exact sequence of homotopy groups of the fibration 
in Lemma (\ref{R2FibLemma}), we see that we have a sequence
$0 \to \KZ \to \KZ^2 \to \KZ \to 0.$ 
Now, $\pi_1(R_2)$ acts on each term of this sequence by 
$\KZ$-module automorphisms with the
trivial action on the first and last $\KZ$ factors and by an
action $\theta$ on the middle factor.

This implies that $\theta$ may be taken to be the homomorphism
induced by the matrix
$$\left( 
\begin{array}{cc}
 1 & \ast \\
 0 & 1 \\
\end{array}
\right).
$$

We claim that $$\theta \simeq \left( \begin{array}{cc}
1 & 1 \\
0 & 1 \\
\end{array}
\right).
$$

We have a fibration
$K(\KZ,2) \to R_2 \to K(\KZ,2) \times K(\KZ,1).$
We could view this as a fibration over the $K(\KZ,1) \simeq S^1$
factor 
$K(\KZ^2,2) \to R_2 \to S^1.$
From the long exact sequence of a fibration, it is clear that
$K(\KZ^2,2)$ is the universal cover $\tilde{R_2}$ of $R_2.$
Now $\pi_1(S^1)$ acts on $\tilde{R_2}$ by
deck transformations. Hence, using the
Cartan-Leray spectral sequence
(see Ref.\ \cite{CartanEilenberg} Ch. XVI Sec.\ (9) for details), 
we have a spectral sequence with 
$E^{p,q}_2 = H^p(\KZ = \pi_1(S^1), H^q(K(\KZ^2,2),\KZ)) \Rightarrow 
H^{\ast}(R_2,\KZ)$ (here $H^{\ast}(\KZ,M)$ denotes the group
cohomology of $\KZ$ with coefficients in a module $M$).
This sequence collapses at the $E_2$ term itself, since 
$E^{p,q}_2 \simeq 0$ for $p \geq 2.$
Thus $\KZ \simeq H^2(R_2,\KZ) \simeq H^0(\KZ,H^2(K(\KZ^2,2),\KZ)) \simeq 
H^0(\KZ,\KZ^2),$  and so the fixed points of $\theta$ on
$\KZ^2$ are $\KZ \neq \KZ^2.$ Hence the action $\theta$ is {\em not}
trivial and $R_2$ is {\em not} simple.

Now, 
$\KZ \simeq H^3(R_2,\KZ) \simeq H^1(\KZ,H^2(K(\KZ^2,2),\KZ)) 
\simeq H^1(\KZ,\KZ^2).$ 
If $\KZ$ acts on $\KZ^2$ with an action $\theta,$ 
$H^{\ast}(\KZ,\KZ^2)$ is the cohomology of the complex
$ \KZ^2 \overset{\theta -1}{\to} \KZ^2.$ 
Hence, $\KZ^2/(\theta -1 ) \KZ^2 \simeq \KZ$ here,
and using the above form for $\theta,$ 
$$ \theta \simeq 
\left(
\begin{array}{cc}
1 & 1 \\
0 & 1 \\
\end{array}
\right).
$$

\end{proof}

Note that since $R_2$ was defined in the unbased category,
for any CW-complex $W,$ pairs over $W$ are classified by unbased
maps from $W$ to $R_2$. Since the space of unbased maps from
$W$ to $R_2$ is the quotient of the space of based maps from
$W$ to $R_2$ by the action of $\pi_1(R_2)$, we see that the 
non-trivial action of $\pi_1(R_2)$ does not affect our results.
We simply have to be careful to use unbased maps in all our constructions.
We can see an example of this when we try to determine all the pairs 
over $S^2$. If pairs were classified by based maps, then 
$P_2(S^2)$ would be a group. However, we have the following

\begin{lemma}
$P_2(S^2)$ is not a group.
\label{PairsOnS2}
\end{lemma}
\begin{proof}
For any CW complex $W,$ we have a natural map $\phi:P_2(W) \to H^2(W,\KZ)$
which sends a pair $([p],b)$ over $W$ to $[p].$
Now, for every $a \in H^2(S^2,\KZ) \simeq \KZ,$
we claim that the set $\phi^{-1}(a)$ has cardinality $|a|$.
To see this, it is enough to note that if $D_p \to S^2$ is a 
principal $S^1$-bundle of Chern class
$[p],$ then $H^2(D_p,\KZ) \simeq \KZ_p$.
This implies that $P_2(S^2) \to H^2(S^2,\KZ)$ is not
a group homomorphism, and hence that $P_2(S^2)$ is not a group.
\end{proof}

In fact, $P_2(S^2)$ is the quotient of the group $\pi_2(R_2) \simeq \KZ^2$ by
the action of $\pi_1(R_2)$ calculated in
Thm.\ (\ref{ThmR2Action}) above.

\begin{lemma}
Consider the principal $S^1$-bundle $p:E_2 \to R_2$ whose characteristic class 
is $b$ the generator of $H^2(R_2,\KZ).$ Its cohomology groups are
\begin{itemize}
\item $H^0(E_2,\KZ) \simeq \KZ,$
\item $H^1(E_2,\KZ) \simeq \KZ \bar{a},$
\item $H^2(E_2,\KZ) \simeq \KZ w,$
\item $H^3(E_2,\KZ) \simeq \KZ \bar{c}.$
\end{itemize}
where $\bar{a}, \bar{c}$ are the images of the generators $a, c$ of 
$H^i(R_2,\KZ), i=1,3$ in $H^{\ast}(E_2,\KZ)$. In addition $p_!(w) = a$,
$\bar{a}=  p^{\ast}(a), \bar{c} = p^{\ast}(c).$
\end{lemma}
\begin{proof}
The Serre Spectral Sequence table is 
\begin{equation}
\begin{array}{c|ccc}
\KZ z & \KZ az & \KZ bz & \KZ cz \\
&  &  & \\
\hline
\KZ & \KZ a & \KZ b & \KZ c \\
\end{array}
\end{equation}
The transgression map $d^2: E_2^{0,1} \to E_2^{2,0}$ has to be $dz = b.$
Hence using the ring structure,
the fact that $ab=0$ in $H^{\ast}(R_2,\KZ),$
and the fact that $b$ generates a copy of $\KZ[b]$ inside $H^{\ast}(E_2,\KZ)$
gives the cohomology groups shown.
In addition, note that $w$ is the image of $a.z$ in
$H^2(E_2,\KZ).$ 

Then, an inspection of the Gysin sequence for $E_2$ shows that
the elements $\bar{a},\bar{c}$ must be images of $a,c$ under
$p^{\ast}.$
\end{proof}
This bundle acts as a universal bundle for pairs: 
Given a principal $S^1$-bundle $p:D \to W$ and $b \in H^2(D,\KZ)$ we have
a classifying map $f:W \to R_2.$ This bundle pulls back along
$f$ to give the bundle $p:D \to W$
while the generator of $H^2(E_2,\KZ)$ pulls back to the $b$-class.

%
%
%

We hope to study the map $T$ of Section $(1)$ using the
classifying space $R_2$ studied above.

\section{T-duality for automorphisms is not involutive \label{SecInvol}}
By the proof of Thm.\ (3.1) in Ref.\ \cite{ATMP}, we know that the T-dual
of an automorphism even with $H$-flux is always unique. However,
T-duality for automorphisms is not involutive: If we perform two 
successive T-dualities
we may not get the automorphism we started with.
For example, if $W=S^2$ with $1$ unit of $H$-flux on $S^2 \times S^1$, 
the T-dual is $S^3$ with no $H$-flux. Since 
$H^2(S^2 \times S^1,\KZ) \simeq \KZ,$ but $H^2(S^3,\KZ) \simeq 0$,
every locally unitary (but not necessarily unitary) automorphism 
of $CT(S^2\times S^1, 1)$ dualizes 
to a unitary automorphism of $C(S^3,\Cpct)$. 
From the proof of Thm.\ (3.1) of Ref.\ \cite{ATMP}, it is clear that the
T-dual of a unitary automorphism is unitary. Hence, taking one more T-dual
gives a unitary automorphism of $CT(S^2 \times S^1,1)$.

At the level of triples $([p],b,H)$ we conjecture that the T-duality
map of Section (1) should have the following properties:
\begin{lemma}
Let $W$ be connected and simply connected.
Let $p:E \to W$ be a principal $S^1$-bundle with $H$-flux $H$ and
$b \in H^2(E,\KZ).$
Let $q:E^{\#} \to W$ be the T-dual principal $S^1$-bundle with 
$H$-flux $H^{\#}$ and $b^{\#} \in H^2(E^{\#},\KZ)$ where $b^{\#} = T(b).$ 
Then, for all $b \in H^2(E,\KZ), \forall l,m \in \KZ$, 
the Gysin sequence induces a natural isomorphism
$H^2(E/\KZ)/<p^{\ast}p_!(H)> \simeq H^2(E^{\#},\KZ)/<q^{\ast}q_!(H^{\#})>.$
\end{lemma}
\begin{proof}
The Gysin sequence of $q:E^{\#} \to W$ is
$$\KZ \overset{[q]}{\to} H^2(W,\KZ) \overset{q^{\ast}}{\to} 
H^2(E^{\#},\KZ)\overset{q_!}{\to} H^1(Z,\KZ) \to \cdots$$
The Gysin sequence of $p:E \to W$ is
$$\KZ \overset{[p]}{\to} H^2(W,\KZ) \overset{p^{\ast}}{\to} H^2(E,\KZ)
\overset{p_!}{\to} H^1(Z,\KZ)\to\cdots$$
Since $H^1(W,\KZ)=0$, $p^{\ast},q^{\ast}$ are surjective by 
the Gysin sequence. Suppose $b=p^{\ast}(a)$ and $b^{\#}=q^{\ast}(a^{\#})$
for $a,a^{\#} \in H^2(W,\KZ)$.

The kernel of $p^{\ast}$ is the subgroup\footnote{
$\langle a_1,a_2,\cdots,a_n \rangle$ 
denotes the subgroup generated by $a_1, a_2, \cdots, a_n$. The ambient
group is understood from context.}
$\langle [p] \rangle $. 
Similarly, the kernel of $q^{\ast}$ is the subgroup 
$\langle [q] \rangle$. Let $G = \langle [p],[q] \rangle = 
\langle p_!(H), q_!(H^{\#}) \rangle$.
Since $H^1(W,\KZ)=0$, $p^{\ast},q^{\ast}$ are surjective by 
the Gysin sequence. Note that $p^{\ast}G \simeq \langle p^{\ast}p_!(H) \rangle$
and $q^{\ast}G \simeq \langle q^{\ast} q_!(H^{\#}) \rangle$.
Thus, we have isomorphisms 
$H^2(W,\KZ)/G \simeq H^2(E,\KZ)/\langle p^{\ast}p_!(H) \rangle 
\simeq H^2(E^{\#},\KZ)/\langle q^{\ast}q_!(H^{\#})\rangle$.
\end{proof}

The cosets corresponding to $b$ and $b^{\#}$ may be written
as ($l,m \in \KZ$)
$$
\{b + l p^{\ast}p_!(H)\} \protect{\mbox{ and }}
\{b^{\#} + m q^{\ast}{q_!(H^{\#})}\}.
$$

{\em We conjecture that each coset is precisely the collection of 
$b$-fields with the same T-dual (even when $X$ is {\bf not }
simply connected). }
As support for this, note the following:
Suppose $E^{\#} = W \times S^1,$ with $k$ units of $H$-flux. Let 
$q=\pi: W \times S^1 \to W$ be the projection
map. Then, if $H^1(W,\KZ) \neq 0,$ 
the above theorem would not be expected to hold: 
For one thing, $\protect{\mbox{im}}(\pi^{\ast})$
would not be all of $H^2(W\times S^1,\KZ)$.
It is strange then, that $H^2(E,\KZ)/\langle p^{\ast}p_!(H) \rangle$
is isomorphic to $H^2(E^{\#},\KZ)/\langle q^{\ast}q_!(H^{\#}) \rangle$
in all the following cases in many of which $H^1(W,\KZ) \neq 0$ 
(I use Ref.\ \cite{BouEvMa} for the examples):

\begin{enumerate}
\item \underline{$W=T^2:$} 
For $E^{\#} = W\times S^1,$ $H^0(E^{\#},\KZ)=\KZ, H^1=\KZ^3, H^2 = \KZ^3,
H^3 = \KZ$. The $H$-flux is a class 
$[p] \times z \in H^2(T^2) \otimes H^1(S^1) \simeq H^3(T^2 \times S^1) 
\simeq \KZ$.
The T-dual is the nilmanifold $p:N \to T^2$ whose cohomology is
$H^0=\KZ, H^1=\KZ^2, H^2 = \KZ^2 \oplus \KZ_p,H^3=\KZ$. 
It is clear that $\KZ^3 / p\KZ \simeq \KZ^2 \oplus \KZ_p$.

\item \underline{$W=M,$ an orientable surface of genus $g>1:$} 
The  cohomology of
$W \times S^1$ is $H^0=\KZ,H^1=\KZ^{2g+1},H^2=\KZ^{2g+1},H^3=\KZ$. The
$H$-flux is a class $j \times z \in H^3 \simeq H^2(M)\otimes H^1(S^1) 
\simeq \KZ$.
The T-dual is a bundle $j:E\to M$ with 
$H^0=\KZ,H^1=\KZ^{2g}, H^2=\KZ^{2g}\oplus \KZ_j,H^3=\KZ$. 
Here, $\KZ^{2g+1}/j \KZ \simeq \KZ^{2g} \oplus \KZ_j$.

\item \underline{$W={\mathbb R}{\mathbb P}^2:$}
The cohomology of $W\times S^1$ is
$H^0=\KZ,H^1=\KZ,H^2=\KZ_2\simeq H^2(W) \otimes H^0(S^1), H^3 \simeq
H^2(W)\otimes H^1(S^1) \simeq \KZ_2$. The $H$-flux is the class
$1\times z \in H^2(W) \otimes H^1(S^1)$.  
The T-dual is a bundle $k:E\to {\mathbb R}{\mathbb P}^2 $ with 
$H^0= \KZ,H^1= \KZ,H^2=0,H^3=\KZ_2$.
Once again, $\KZ_2 /1 \KZ_2 \simeq 0$.

\item \underline{$W={\mathbb R}{\mathbb P}^3:$}
The cohomology of $W\times S^1$ is
$H^0=\KZ,H^1=\KZ,H^2=\KZ_2,H^3=\KZ \oplus \KZ_2,H^4=\KZ$. Then, 
$H^3 \simeq H^2(W) \otimes H^1(S^1) \oplus H^3(W)\otimes H^0(S^1)$.
The $H$-flux is $1\times z + k \times 1$. 
Now, $\pi^{\ast} \pi_{!}(H) = 1 \times 1$.
The T-dual is $q:S^1 \times S^3 \to {\mathbb R}{\mathbb P}^3$ with
cohomology $H^0=\KZ,H^1=\KZ,H^3=\KZ,H^4=\KZ$. The T-dual
has no $B$-class and $H$-flux $k \in \KZ$.

\item \underline{$W= {\mathbb R}{\mathbb P}^{2m} (m > 1):$}
The cohomology of $W \times S^1$
is $H^0=\KZ,H^1=\KZ,H^q=\KZ_2,q=2,\cdots,m-1,H^{2m}=\KZ,H^{2m+1}=\KZ_2$.
The $H$-flux is the class 
$1 \times z \in H^2({\mathbb R}{\mathbb P}^{2m}) \otimes H^1(S^1)$. 
The T-dual is a bundle $q:E \to {\mathbb R}{\mathbb P}^{2m}$ 
with cohomology $H^0=\KZ,H^1=\KZ,H^{2m+1}=\KZ_2$.
Here, $\KZ_2/1 \KZ_2 \simeq 0$.

\item \underline{$W= {\mathbb R}{\mathbb P}^{2m+1} (m > 1):$}
The cohomology of $W \times S^1$
is $H^0=\KZ,H^1=\KZ,H^q=\KZ_2,q=2,\cdots, m-1, H^{2m+1}=\KZ\oplus \KZ_2,
H^{2m+2}=\KZ$. Note $H^3=\KZ_2;$ we have
that $H^3 \simeq H^2({\mathbb R}{\mathbb P}^{2m+1}) \otimes H^1(S^1)$.
The $H$-flux is the class $1\times z \in H^3 \simeq \KZ_2$.
The T-dual is a principal bundle 
$q:S^1 \times S^{2m+1} \to {\mathbb R}{\mathbb P}^{2m+1}$ with
cohomology $H^0=\KZ,H^1=\KZ,H^{2m+1}=\KZ,H^{2m+2}=\KZ$. 
The T-dual has no second cohomology as expected.

\item \underline{$W= \KCP^2:$} 
The cohomology of $W\times S^1$ is 
$H^0=\KZ,H^1=\KZ,H^2=\KZ,H^3=\KZ,H^4=\KZ,H^5=\KZ$. We have
$H^3 \simeq H^2(W)\oplus H^1(S^1) \simeq \KZ$. The $H$-flux is
the class $j\times z \in H^3$. 
The T-dual is the Lens space $L(2,j) \to \KCP^2$ if $j\neq 0$. 
It has cohomology $H^0=\KZ,H^2=\KZ_j,H^4=\KZ_j, H^5=\KZ$.
Once again, $\KZ/j \KZ \simeq \KZ_j$.
\end{enumerate}

\section{Properties of $R_{3,2}$}
We noted in Sec.\ (\ref{SecBr}) that from an element of ${\mathcal S}$ we
may obtain a {\em triple} over $W$ consisting of a principal 
$S^1$-bundle $p:E \to W$
a class $b \in H^2(E,\KZ)$ and a class $H \in H^3(E,\KZ).$
As in Sec.\ (\ref{SecClass}), the map
$W \to (\text{\bf Triples over } W)$ is a set-valued functor on
the category of unbased CW complexes. It has a classifying space
denoted $R_{3,2}.$ We have natural fibrations $q_2:R_{3,2} \to R_2$
and $q_3:R_{3,2} \to R_3.$ We will show below that this classifying
space has a canonical bundle over it which classifies triples.

We noted above that $R_{3,2}$ is a fiber product 
$R_2 \underset{K(\KZ,2)}{\times} R_3.$ 
As in Sec.\ (\ref{SecClass}), this classifying space is unbased. However,
one may choose a basepoint for it as in that section. 
We pick basepoints $r_i \in R_i$ and a basepoint $r_{3,2} \in R_{3,2}$
such that $q_2(r_{3,2}) = r_2$ and $q_3(r_{3,2}) = r_3.$ (That is, the
$q_i$ may be taken to be based maps.)

\begin{theorem}
\begin{enumerate}
\item The homotopy groups of $R_{3,2}$ are
\begin{itemize}
\item $\pi_1(R_{3,2}) \simeq \KZ,$
\item $\pi_2(R_{3,2}) \simeq \KZ^3,$
\item $\pi_3(R_{3,2}) \simeq \KZ,$
\item $\pi_i(R_{3,2}) \simeq 0,$  $i > 3.$
\end{itemize}
\item $R_{3,2}$ is not simple
\end{enumerate}
\label{UnivCover}
\end{theorem}
\begin{proof}
Given an unbased map $f:W \to R_{3,2},$ classifying a triple 
$([p],b,H)$ over $W,$ we pick a basepoint $x_0$ for $W.$ 
We use the argument of Lemma (\ref{R2BpLemma}), to homotope $f$ to 
a based map. Note that
$q_3 \circ f$ is a pair over $W,$ the pair $([p],H).$ If 
$q_3 \circ f$ is nullhomotopic, then $[p] = 0, H=0$ and hence
$b$ defines classes in $H^2(W \times S^1,\KZ)$ and 
$H^1(W \times S^1,\KZ).$
That is, we obtain a based map $W \to K(\KZ,2) \times K(\KZ,1)$
and hence the homotopy fiber of $q_3$ is $K(\KZ,2) \times K(\KZ,1).$
Similarly we see that the homotopy fiber of $q_2$ is
$K(\KZ,3) \times K(\KZ,2).$

\begin{enumerate}
\item We know that $R_{3,2}$ may be taken to be a based fibration of the form
$$ K(\KZ, 2) \times K(\KZ,1) \to R_{3,2} \to R_3$$
and also a based fibration of the form
$$ K(\KZ,3) \times K(\KZ,2) \to R_{3,2} \to R_2.$$
This gives two long exact sequences for the homotopy groups of 
$R_{3,2}:$
$0 \to \pi_3 \to \KZ \to \KZ \to \pi_2 \to \KZ^2 \to \KZ \to \pi_1 \to 0$
and 
$0 \to \KZ \to \pi_3 \to 0 \to \KZ \to \pi_2 \to \KZ^2 \to 0 \to \pi_1 \to \KZ.$
From these sequences it is clear that $\pi_i(R_{3,2}) \simeq 0$ if
$i > 3.$ Further, from the second sequence $\pi_3 \simeq \KZ,$
and $\pi_2 \simeq \KZ^3.$ Then it follows that $\pi_1 \simeq \KZ.$
\item From the long exact sequence of the fibration over $R_2,$ we have
a $\pi_1(R_2)$ equivariant sequence 
$0 \to \pi_2(K(\KZ,3) \times K(\KZ,2)) \to \pi_2 \to \KZ^2 \to 0 \to
\pi_1 \to \KZ.$
From this, it is clear that the generator of $\pi_1(R_{3,2})$ maps 
isomorphically to the generator of $\pi_1(R_2).$ In addition, this generator
acts on the degree two part of the sequence. Since 
$K(\KZ,3) \times K(\KZ,2)$ is simply connected and since
every action of $\pi_1$ factors through the action of
the fundamental group, the generator of $\pi_1(R_{3,2})$
acts trivially on $\pi_2(K(\KZ,3) \times K(\KZ,2))$. 
In the above sequence, this generator acts nontrivially on 
$\pi_2(R_2) \simeq \KZ^2$
and hence must act nontrivially on $\pi_2(R_{3,2})$ since the 
sequence is $\pi_1$-equivariant. Hence $R_{3,2}$ is not simple. 
\end{enumerate}
\end{proof}
We determine the action of $\pi_1(R_{3,2})$ on $\pi_2(R_{3,2})$ below
(see Eq.\ (\ref{Pi2OnR2}) below and surrounding text).

The T-duality map should be a map $T_{3,2}:R_{3,2} \to R_{3,2}.$ It might
be hoped that this map is a fiber product of two maps
$R_3 \to R_3$ and $R_2 \to R_2.$ Unfortunately, $T_{3,2}$ cannot be of
this form: Note that under T-duality a pair $([p],b)$ maps into
a {\em triple} $(0,b^{\#},[p] \times z).$ 
%
%
To determine $T,$ we first need to study the structure of $R_{3,2}.$
Since $R_{3,2}$ is a fibration over $R_2,$ and $R_2$ is a mapping torus
(see previous section), it might be hoped that $R_{3,2}$ is also a
mapping torus. This is indeed the case.

\begin{lemma}
The universal cover $\tilde{R}_{3,2}$ of $R_{3,2}$ is homotopy equivalent
to $R_3 \times K(\KZ,2).$
Hence $R_{3,2}$ is a mapping torus of a map 
$\psi: R_3 \times K(\KZ,2) \to R_3 \times K(\KZ,2).$
\end{lemma}
\begin{proof}
Let $p:\tilde{R}_{3,2} \to R_{3,2}$ be the covering projection. 
Pick $\tilde{r}_{3,2} \in \tilde{R}_{3,2}$ 
such that $p(\tilde{r}_{3,2}) = r_{3,2}.$
Consider the map $\phi_{\ast}:\pi_2(\tilde{R}_{3,2}) \to \pi_2(R_3) $
given by $\phi_{\ast} = q_{2\ast} \circ p_{\ast}.$
$\pi_2(\tilde{R}_{3,2}) \underset{p_{\ast}}{\to}
\pi_2(R_{3,2}) \underset{q_{3 \ast}}{\to} \pi_2(R_3).$
Note that $p_{\ast}$ is an isomorphism by definition. I claim $\phi$ is onto:
We have a fibration $K(\KZ,2) \times K(\KZ,1) \to R_{3,2} \overset{q_3}{\to}
R_3.$ 
The long exact sequence of homotopy groups associated to this fibration
is  

\begin{equation}
\CD
0 @>>> \KZ @>\lambda>> \KZ @>\mu>> \KZ @>\kappa>> \KZ^3 \\
@>\phi>> \KZ^2 @>>> \KZ @>>> \KZ @>>> 0. \\
\endCD
\label{EqSeq1} \\
\end{equation}
Note that all the groups in the sequence are torsion-free.
Here, $\lambda$ is an injective map $\KZ \to \KZ$ and hence
$\lambda$ is an isomorphism and $\ker(\mu) \simeq \KZ.$ Hence by exactness
and the absence of torsion $\mu = 0,\kappa$ is injective and $\phi$ onto. 

Consider the map $\phi=q_3 \circ p:\tilde{R}_{3,2} \to R_3.$ 
By the arguments in Lemma (\ref{R2BpLemma})
we may take it to be based. Let $W$ be its homotopy fiber.
Then the l.e.s. of the fibration $W \to \tilde{R}_{3,2} \to R_3$ is

\begin{equation}
\CD
\minCDarrowwidth1pc
0 @>>> \pi_3(W) @>\alpha>> \KZ @>\beta>> \KZ @>\gamma>>\pi_2(W) \\
@>\mu>> \KZ^3 @>\kappa>> \KZ^2 @>\nu>> \pi_1(W) @>>> 0 @>>> 0.
\endCD
\nonumber \\
\end{equation}

We know $\kappa$ is onto and hence $\nu = 0$ and $\pi_1(W)=0.$
Since $\kappa$ is onto, $\ker(\kappa) \simeq \KZ,$ and hence, 
$\im(\mu) \simeq \KZ.$ 

Now, consider the map $\phi_{\ast}:\pi_3(\tilde{R}_{3,2}) \to \pi_3(R_3) $
given by $\phi_{\ast}= q_{3\ast} \circ p_{\ast}$ in degree $3.$
Note that $p_{\ast}$ is an isomorphism by definition.
By the l.e.s. Eq.\ (\ref{EqSeq1}), $\beta = \lambda \circ p_{\ast}.$ 
By the argument given above both $\lambda$ and $p_{\ast}$ are isomorphisms,
and so is $\beta.$
Hence, $\im(\alpha) = 0$ and $\pi_3(W) \simeq 0.$

Also $\ker(\gamma) \simeq \KZ$ and so $\gamma = 0.$
Hence $\pi_2(W) \simeq \KZ$ by counting ranks.
Its clear that $\pi_i(W) \simeq 0$ if $i > 3.$
Therefore, $W$ is homotopy equivalent to $K(\KZ,2).$

Now, $W \to \tilde{R}_{3,2} \to R_3$ is a based fibration with 
$W \sim K(\KZ,2).$  Also, $R_{3,2}$ and $R_3$ are spaces of finite 
type since their fundamental groups are finitely generated 
(by Ref.\ \cite{Wall}) and the action of $\pi_1(R_3)$ on $W$ is
zero since $\pi_1(R_3) \simeq 0$.  By Lemma $8^{bis}.28$ of 
Ref.\ \cite{McCleary}, this implies that the fibration is principal.
%
%
Hence, $\tilde{R}_{3,2} \simeq K(\KZ,2) \times R_3$  because by 
Ref.\ \cite{Bunke}
$H^3(R_3,\KZ) = 0$ and so the classifying map $R_3 \to K(\KZ,3)$ is
always trivial.

Since $R_3 \times K(\KZ,2)$ is homotopy equivalent to 
the universal cover of $R_{3,2},$ we have 
a commutative diagram of spaces
\begin{equation}
\CD
R_3 \times K(\KZ,2) @>{\tilde{q}_2 }>> K(\KZ,2) \times K(\KZ,2)\\
@V{p_{3,2}}VV  @V{p_2}VV  \\
R_{3,2} @>q_2>> R_2 @>>> K(\KZ,1).
\endCD
\nonumber
\end{equation}
which gives a diagram of homotopy groups and maps equivariant under
the action of $\pi_1(K(\KZ,1))\simeq \KZ$. It is clear that since $R_2$ is the
mapping torus of $K(\KZ^2,2)$ under the lift of this action, 
$R_{3,2}$ is the mapping torus of $R_3 \times K(\KZ,2)$ under the
lift of this action (by the commutativity of the diagram).
The generator of this action acts on $R_3 \times K(\KZ,2)$ as the
generator $\psi$ of the deck transformation group. 
\end{proof}

It is interesting to note that we have a sequence of 
mapping tori
$$
R_{3,2} \to R_2 \to K(\KZ,1)
$$
such that there are isomorphisms
$\pi_1(R_{3,2}) \simeq \pi_1(R_2) \simeq \pi_1(K(\KZ,1))$ and 
$\pi_1(K(\KZ,1))$ acts on the homotopy groups of each of these spaces.

Note that $\tilde{R}_2$ is a trivial $K(\KZ,2)$ fibration over 
$K(\KZ,2)$ (since $R_2$ is).
Also, $\tilde{R}_{3} \simeq R_3$ is a fibration over $K(\KZ,2)$ as
well (since $R_3$ is). It then follows from the definition 
that $\tilde{R}_{3,2}$ is the fiber
product of $\tilde{R}_3$ and $\tilde{R}_2$ over $K(\KZ,2).$ 

We have a commutative diagram
\begin{equation}
\CD
R_3 \times K(\KZ,2) @>{\tilde{q}_2 }>> K(\KZ,2) \times K(\KZ,2)\\
@V{p_{3,2}}VV  @V{p_2}VV \\
R_{3,2} @>q_2>> R_2.
\endCD
\nonumber
\end{equation}

We note that $\pi_2(R_{3,2}) \simeq \pi_2(\tilde{R}_{3,2}).$
We choose generators for $\pi_2(\tilde{R}_{3,2}) \simeq \KZ^3$ 
such that the projection
$(a,b,c) \mapsto c$ is induced by the map
$R_3 \times K(\KZ,2) \to K(\KZ,2).$
Also the projection $(a,b,c) \mapsto (a,b)$ is induced by the
map $R_3 \times K(\KZ,2) \to R_3.$
With these choices of generators, the map $\tilde{q}_2$ 
is given by $(a,b,c) \mapsto (a,c).$

In Sec.\ (\ref{SecClass}) we had noted that the deck group of $\tilde{R}_2$
acts on pairs by multiplication by a matrix of the form
$$
\left(
\begin{array}{cc}
1 & 1 \\
0 & 1 \\
\end{array}
\right).
$$

The map $\tilde{q}_2^{\ast}$ on $\pi_2$ 
must be equivariant under the action of
the deck group since 
$q_2^{\ast}$ is equivariant under the action of $\pi_1$
and the maps $p_{3,2}^{\ast}$ and $p_2^{\ast}$ are isomorphisms on 
$\pi_2.$
This implies that the deck group acts on $\pi_2(\tilde{R}_{3,2})$
by multiplication by the matrix 
\begin{equation}
\left(
\begin{array}{ccc}
1 & 0 & 1 \\
0 & 1 & 0 \\
0 & 0 & 1 \\
\end{array}
\right).
\label{Pi2OnR2}
\end{equation}
Using the isomorphisms $p_{3,2}^{\ast}$ and $p_2^{\ast}$ above,
this is also the action of $\pi_1(R_{3,2})$ on $\pi_2(R_{3,2}).$
We may also determine the deck transformation explicitly as follows:

\begin{lemma}
The deck transformation $\psi:\tilde{R}_{3,2} \to \tilde{R}_{3,2}$ 
has the form $\psi=\pi_1 \times f$ where $\pi_1:R_3 \times K(\KZ,2) \to R_3$ is
the projection onto the first factor and $f:R_3 \times K(\KZ,2) \to K(\KZ,2)$
is determined below.
\end{lemma}
\begin{proof}
Consider a triple over a simply connected space $W.$ This determines
a classifying map $f:W \to R_{3,2}.$ Using the argument of 
Lemma (\ref{R2BpLemma}) we may take it to be based.
Since $W$ is simply connected, this
map lifts to $\tilde{R}_{3,2}.$ This implies that a triple $(p,b,H)$
over $W$ is determined by a pair $([p],H)$, $p:E \to W,$ and a class $b$ in 
$H^2(W,\KZ).$ From the Gysin sequence, since $H^1(W,\KZ) \simeq 0$
we see that 
$\KZ \simeq H^0(W,\KZ)\stackrel{\cup p}{\to} H^2(W,\KZ) \to H^2(E,\KZ) \to 0$ 
is exact. Thus, every class in $H^2(E,\KZ)$ is the image of
some class in $H^2(W,\KZ)$ from the Gysin sequence. This is the class
$b$ above. Further two such classes differ by an integral multiple of $p$.

It is also clear that if $\psi:\tilde{R}_{3,2} \to \tilde{R}_{3,2}$ is
a deck transformation, then, since $p_{3,2} \circ \psi=\psi$,  
$\psi \circ f$ represents the same triple over $W.$ Thus, 
$\psi \circ f$ is also representable as $((p,H),b')$ for some
$b'$.
If $\tilde{f}$ is the lift of a map $f:W \to R_{3,2}$ representing a 
triple $(p,b,H)$ to $\tilde{R}_{3,2}$
then it is clear that all possible lifts may be 
represented as $((p,H),y+mp), m \in \KZ$ where
$y \in H^2(W,\KZ)$ and $p^{\ast}(y)=b$. 
By the above, the action of the deck transformation is to 
shift $((p,H),y)$ to $((p,H),y + m p)$. 
Thus, $\psi = \pi_1 \times f$ where 
$f:R_3 \times K(\KZ,2) \to K(\KZ,2)$ is the map defined by
$((p,H),y) \to y + mp, m \in \KZ$, and $\pi_1:R_3 \times K(\KZ,2) \to R_3$
be the projection onto the first factor. 
\end{proof}

It is clear that the map $f$ above defines a 
class in $H^2(R_3 \times K(\KZ,2), \KZ)$. 

As an illustration of this, consider all triples over $S^2.$ 
For any principal circle bundle $E_p \to S^2,$ we have that 
$H^3(E_p,\KZ) \simeq \KZ.$ Hence,
for each pair $([p],b),$ the $H$-flux can have countably many values.

\begin{lemma}
The set of unbased homotopy classes of
maps $S^2 \to R_3 \times K(\KZ,2)$ modulo the action of the
the deck transformation group on $\tilde{R}_{3,2}$
is exactly the set of triples over $S^2.$
\end{lemma}
\begin{proof}
Clearly, the set of unbased homotopy classes of
maps $S^2 \to R_3 \times K(\KZ,2)$ 
may be made based since both $R_3$ and $K(\KZ,2)$ are connected and
simply connected: This set is exactly $\pi_2(R_3 \times K(\KZ,2)).$

Hence, this set modulo the action of the
the deck transformation group on $\tilde{R}_{3,2}$
is the same as $\pi_2(\tilde{R}_{3,2})$ modulo the action of 
$\pi_1(\tilde{R}_{3,2}).$ By Thm.\ (\ref{UnivCover}), this is the
quotient of $\KZ^3$ under the action of the matrix 
$$
\left(
\begin{array}{ccc}
1 & 1 & 0 \\
0 & 1 & 0 \\
0 & 0 & 1 \\
\end{array}
\right).
$$

From the previous section (in particular the discussion before and after
Thm.\ (\ref{PairsOnS2}))
we know that the quotient of $\KZ^2$ by
the matrix
$$
\left(
\begin{array}{cc}
1 & 1 \\
0 & 1 \\
\end{array}
\right)
$$
is exactly the collection of pairs over $S^2.$ The action on $\KZ^3$
leaves the last coordinate fixed. Hence, for each isomorphism
class of pairs $([p],b)$ over $S^2,$ we obtain a countable number
of triples as required. 
\end{proof}
Note that at least for $S^2,$ the only action of the deck group is
to shift a pair $([p],H)$ to an equivalent pair. 

\section{The T-duality Mapping \label{SecTriple}}
We note that the maps $q_2:R_{3,2} \to R_2$ and $q_3:R_{3,2} \to R_3$
have natural sections. Thus these maps are injective on cohomology.
In particular (from the Theorem below) 
$H^2(R_{3,2},\KZ) \simeq \KZ^2 \simeq H^2(R_3,\KZ) \simeq 
\KZ a_1 \oplus \KZ a_2.$ 
Similarly, $H^2(R_2,\KZ) \into H^2(R_{3,2},\KZ).$ The latter map may be
taken as the inclusion $\KZ \into \KZ \oplus \KZ$ into the first factor.
Thus we may write $a_1 = q_2^{\ast}(b).$
Further, from the Theorem below, we see that $H^1(R_{3,2},\KZ) \simeq \KZ l.$ 
Therefore $l = q_2^{\ast}(a)$ and so $a_1 \cdot l=0$
since $b \cdot a =0$ in $H^{\ast}(R_2,\KZ)$ (see Eq.\ (\ref{R2Rel})).

In addition since $R_{3,2} \to R_3$ is fiber-preserving over 
$K(\KZ,2),$ and the natural map $R_{3,2} \to K(\KZ,2)$ has 
sections, we see that 
$\KZ[c] \simeq H^{\ast}(K(\KZ,2),\KZ) \into H^{\ast}(R_{3,2},\KZ).$

\begin{theorem}
The cohomology groups of $R_{3,2}$ are 
\begin{itemize}
\item $H^0(R_{3,2},\KZ) \simeq \KZ,$
\item $H^1(R_{3,2},\KZ) \simeq \KZ l,$
\item $H^2(R_{3,2},\KZ) \simeq \KZ a_1 \oplus \KZ a_2,$
\item $H^3(R_{3,2},\KZ) \simeq \KZ a_2 l,$
\item $H^4(R_{3,2},\KZ) \simeq \KZ a_1^2 \oplus \KZ a_2^2 \oplus \KZ x$
\end{itemize}
where $a_1 \cdot a_2 = 0$ and $a_1 \cdot l = 0.$ Also, $a_1 = q_2^{\ast}(b).$
\end{theorem}
\begin{proof}

$R_{3,2}$ is connected, being a fiber product of two connected spaces
over a connected space. Hence $H^0(R_{3,2},\KZ) \simeq \KZ.$
We know that the universal cover of $R_{3,2}$ is $R_3 \times K(\KZ,2)$
with deck transformation group $\KZ.$ 
Note that by the Hurewicz theorem, 
$H^2(\tilde{R}_{3,2},\KZ) \simeq \pi_2(\tilde{R}_{3,2}).$ 
Further, since $\tilde{R}_{3,2}$ is the universal cover, 
$\pi_2(\tilde{R}_{3,2}) \simeq \pi_2(R_{3,2}) \simeq \KZ^3.$ 
We calculated the action of $\pi_1(R_{3,2})$ on $\pi_2(R_{3,2})$ above.
Using this we see that the action of $\pi_1(R_{3,2})$ on 
$H^2(\tilde{R}_{3,2},\KZ)$ is multiplication by the same matrix.
The cohomology ring of $\tilde{R}_{3,2}$ is given by
\begin{itemize}
\item $H^0(\tilde{R}_{3,2},\KZ) \simeq \KZ,$
\item $H^1(\tilde{R}_{3,2},\KZ) \simeq 0,$
\item $H^2(\tilde{R}_{3,2},\KZ) \simeq (\KZ a_1 \oplus \KZ a_2) \oplus \KZ c,$
\item $H^3(\tilde{R}_{3,2},\KZ) \simeq 0,$
\item $H^4(\tilde{R}_{3,2},\KZ) \simeq \KZ a_1^2 \oplus \KZ a_2^2$
$\oplus \KZ a_1 c \oplus \KZ a_2 c \oplus \KZ c^2,$
\end{itemize}
The action of the deck group on $H^2$ may be used to calculated
the action of the deck group on $H^4.$ This is multiplication
by the matrix (in the basis $(a_1^2, a_2^2,a_1 c, a_2 c, c^2)$)
$$
\phi = \left(
\begin{array}{ccc|cc}
1 & 0 & 2 & 0 & 1 \\
0 & 1 & 0 & 0 & 0\\
0 & 0 & 1 & 0 & 1\\
\hline
0 & 0 & 0 & 1 & 0\\
0 & 0 & 0 & 0 & 1\\
\end{array}
\right).
$$
We apply the Cartan-Leray
spectral sequence to this universal cover to compute the cohomology 
groups of $R_{3,2}.$ Since $\KZ$ is cohomologically one dimensional,
the sequence collapses at the $E_2$ page itself. 
In particular only $E_2^{q,1}$ and $E_2^{q,2}$ are nonzero.
We find that 
$H^0(\KZ,H^q(\tilde{R}_{3,2},\KZ)) \simeq H^q(\tilde{R}_{3,2},\KZ)^{\KZ}$ 
and
$H^1(\KZ,H^q(\tilde{R}_{3,2},\KZ)) \simeq 
H^q(\tilde{R}_{3,2},\KZ)/(\phi -1)H^q(\tilde{R}_{3,2},\KZ).$
Using this action we obtain the cohomology groups shown.

Here $a_i$ are the pullbacks of the generator of $H^2(R_3,\KZ)$ 
and so $a_1 \cdot a_2 = 0.$ In addition $l$ is the pullback of the
generator $a$ of $H^1(R_2,\KZ)$ via $q_2^{\ast}.$ 
Also, $a_1 = q_2^{\ast}(b)$ and 
so\footnote{See Eq.\ (\ref{R2Rel}).}
$a_1 \cdot l = q_2^{\ast}(b \cdot a) = 0.$ 
Also, $x$ is a new generator in degree $3.$
\end{proof}

Let $p_3:E_3 \to R_3$ be the classifying bundle $E$ 
of Ref.\ \cite{Bunke} Sec.\ (2.4). 
By Sec.\ (\ref{SecClass}), there is a classifying bundle $p_2:E_2 \to R_2.$
By the isomorphisms on $H^2$ discussed above, it is clear that 
$q_2^{\ast}E_2 \simeq q_3^{\ast}E_3$ and we denote this bundle
by $E_{3,2}.$ This bundle has characteristic class $a_1 \in H^2(R_{3,2},\KZ).$
Let $p:E_{3,2} \to R_{3,2}$ be the bundle map.

\begin{theorem}
The cohomology groups of $E_{3,2}$ are 
\begin{itemize}
\item $H^0(E_{3,2},\KZ) \simeq \KZ,$
\item $H^1(E_{3,2},\KZ) \simeq \KZ y,$
\item $H^2(E_{3,2},\KZ) \simeq \KZ p^{\ast}(a_2) \oplus \KZ b,$
\item $H^3(E_{3,2},\KZ) \simeq \KZ p^{\ast}(a_2 l) \oplus \KZ h.$
\end{itemize}
Here $y=p^{\ast}(l)$ so $p_!(y) = 0,$ $p_!(b) = l$ and $p_!(h) = a_2.$ 
\label{ThmR32}
\end{theorem}
\begin{proof}
Consider the Gysin sequences associated to $E_{3,2},E_3$ and $E_2.$ By
naturality, there are morphisms from the sequences associated to $E_3$ and
$E_2$ to the sequence associated to $E_{3,2}$ induced by the maps $q_2,q_3.$ 
Further, the maps $q_i$ have natural sections and hence 
$q_i^{\ast}$ are injective on cohomology. 
\begin{itemize}
\item{\underline{Degree 0}} $E_{3,2}$ is a fibration of a connected space 
($S^1$) over a connected base ($R_{3,2}$). Hence it is connected and 
$H^0(E_{3,2},\KZ) \simeq \KZ.$
\item{\underline{Degree 1}} From the Serre spectral sequence, we find that 
$H^1(E_{3,2},\KZ) \simeq \KZ y.$ From the Gysin sequence for $p:E_{3,2} \to
R_{3,2}$ we have
$$
\KZ l \stackrel{p^{\ast}}{\to} \KZ y \stackrel{p_!}{\to} \KZ.
$$
Now, since $p^{\ast}$ is injective, by exactness, and the absence of torsion
$y = p^{\ast}(l)$ and $p_!(y) = 0.$
\item{\underline{Degree 2}} From the Gysin sequence associated to 
$p:E_{3,2} \to R_{3,2}$ beginning at $H^0(R_{3,2},\KZ)$ we have
\begin{equation}
\begin{CD}
\KZ @>{\phi_1=\cup a_1}>> \KZ a_1 \oplus \KZ a_2 @>{p^{\ast}}>> 
H^2(E_{3,2},\KZ) @>{p_!}>> \KZ l @>{\cup a_1}>> \KZ a_2 l
\end{CD}
\nonumber
\end{equation}
From the previous part of the sequence $\phi_1$ is injective.
Therefore, by exactness 
$H^2(E_{3,2},\KZ) \simeq \KZ p^{\ast}(a_2) \oplus \KZ b.$
Note that the generator $w$ of $H^2(E_2,\KZ)$ pulls back to $H^2(E_{3,2},\KZ)$
as a generator $b.$ Then, we have $p_!(b) = l \in H^1(R_{3,2},\KZ).$ 
Under pullback by a classifying map $W \to R_{3,2},$
$b$ pulls back to the $B$-class of the triple $([p],b,H)$ being classified. 
\item{\underline{Degree 3}}

Consider the Gysin sequence for $p:E_{3,2} \to R_{3,2}$ 
\begin{equation}
\begin{CD}
\KZ l @>{\cup a_1}>> \KZ a_2 l @>{p^{\ast}}>> H^3(E_{3,2},\KZ) @>{p_!}>> 
\KZ a_1 \oplus \KZ a_2 @>{\cup a_1}>> \KZ a_1^2 \oplus \KZ a_2^2 
\oplus \KZ x.
\end{CD}
\nonumber
\end{equation}
We have that $a_1 \cdot l = 0$ (see Eq.\ (\ref{R2Rel})) hence $p^{\ast}$ is
injective, also the last map has kernel 
$\KZ a_2 \subset \KZ a_1 \oplus \KZ a_2.$  By exactness, 
$H^3(E_{3,2},\KZ) \simeq p^{\ast}(a_2 l) \oplus \KZ h.$  
Note that $H^3(E_{3,2},\KZ)$ contains a $\KZ$-subgroup which is
$\im(\lambda_2) \simeq H^3(E_3,\KZ) \simeq h \KZ$ 
such that $p_!(h) = a_2.$
This class pulls back along the classifying map $W \to R_{3,2}$ 
to the $H$-flux of the triple being classified.
\end{itemize}
\end{proof}
Thus over $R_{3,2}$ we have the {\em canonical triple} 
$([p], b, h)$ corresponding to the bundle $E_{3,2}$ above.

\begin{lemma}
The pullback of the fibration $q_3:R_{3,2} \to R_3$ along the T-duality
map $T_3:R_3 \to R_3$ is $R_{3,2}.$
\label{PBR32}
\end{lemma}
\begin{proof}
We have a commutative diagram
\begin{equation}
\CD
T_{3,2}^{\ast}R_{3,2} @>{ T_{3,2} }>> R_{3,2}\\
@V{q_3}VV  @V{q_3}VV \\
R_3 @>T_3>> R_3
\endCD
\nonumber
\end{equation}
where we define $T_{3,2}$ as the map on $R_{3,2}$ induced by the pullback
and, by definition, 
$$
T_{3,2}^{\ast}R_{3,2} = \{ (x,y) \in R_3 \times R_{3,2} | T_3(x) = q_3(y) \}.
$$

Given $f:W \to T_{3,2}^{\ast} R_3, f(w) = (f_1(w),f_2(w))$ such that
$T_3 \circ f_1 = q_3 \circ f_2,$ we get 
$
T_3 \circ f_1: W \to R_3
$
and
$
q_2 \circ f_2: W \to R_2.
$
Obviously these maps define pairs $([p],b)$ and $([p],H)$ with the
{\em same } $[p].$
Conversely, given two such pairs, we obtain maps
$h_1:W \to R_3, h_2:W \to R_2,$ such that
$\mu \circ h_1 = \mu \circ h_2$ where $\mu:R_i \to K(\KZ,2)$ is
the map sending a pair to the class of the principal bundle in
the pair. We define $g_1=T_3^{-1} \circ h_1: W \to R_3.$ 
By the universal property of the fiber product, there is 
a map $g_2:W \to R_{3,2}.$  Then
$q_3 \circ g_2 = h_1 = T_3 \circ g_1$ by definition.
Hence, $(g_1,g_2)$ give a well-defined map $W \to T_{3,2}^{\ast}R_{3,2}.$
Therefore  $T_{3,2}^{\ast}R_{3,2}$ also classifies triples of the
form $([p],b,H)$ since a $2$-pair and a $3$-pair with the same $[p]$ 
define and are defined by the same triple. Hence, by uniqueness of
the classifying space of triples, we get a homeomorphism 
$\phi:T_{3,2}^{\ast} R_{3,2} \to R_{3,2}.$
\end{proof}

Note that if $X$ is a CW-complex, and $f:X \to R_{3,2}$ is
a classifying map which corresponds a triple $(p,b,H)$ on $X$ the
pair associated to this triple is $(p,H).$ The triple
associated to $T_{3,2} \circ f$ is $(p^{\#},b^{\#},H^{\#})$ and
the associated pair is $(p^{\#},b^{\#}).$ Thus, we see that 
$q_2 \circ T_{3,2} \neq q_2.$ However, we always require 
$q_2 = q_2 \circ \phi,$ i.e., we require that $a_1 = q_2^{\ast}(b)$ in 
$H^{\ast}(R_{3,2},\KZ)$ and in $H^{\ast}(T^{\ast}R_{3,2},\KZ)$ as
well. In addition we require that $a_1 \cdot l = a_1 \cdot a_2 = 0$ 
with no such relation on $a_2$ in both spaces.

Note that $q_3 \circ T_{3,2} = T_3 \circ q_3$
and hence $T_{3,2}^{\ast}$ acts on $H^2(R_{3,2},\KZ)$ by interchanging
$a_1$ and $a_2.$ 

\begin{lemma}
Let $\hat{E}_{3,2}$ be the bundle over $R_{3,2}$ of 
characteristic class $a_2 \in H^2(R_{3,2},\KZ).$
The cohomology groups of $\hat{E}_{3,2}$ are
\begin{itemize}
\item $H^0(\hat{E}_{3,2},\KZ) \simeq \KZ,$
\item $H^1(\hat{E}_{3,2},\KZ) \simeq \KZ \hat{y},$
\item $H^2(\hat{E}_{3,2},\KZ) \simeq \KZ p^{\ast}(a_1),$
\item $H^3(\hat{E}_{3,2},\KZ) \simeq \KZ \hat{h},$
\end{itemize}
and $\hat{p}_!(\hat{y}) = 0, \hat{y} = \hat{p}^{\ast}(l),$
$\hat{p}_!(\hat{h}) = a_1.$
\end{lemma}
\label{HatE32CoHo} 
\begin{proof}
The calculation for $H^0$ and $H^1$ is exactly the same
as for $H^0(E_{3,2},\KZ)$ and $H^1(E_{3,2},\KZ)$ with 
$y$ replaced by $\hat{y}.$
\begin{itemize}
\item{\underline{Degree 2}} Consider the Gysin sequence for
$\hat{E}_{3,2}$ beginning at $H^0(R_{3,2},\KZ)$
\begin{equation}
\begin{CD}
\KZ @>{\phi_1=\cup a_2}>> \KZ a_1 \oplus \KZ a_2 @>{\hat{p}^{\ast}}>> 
H^2(E_{3,2},\KZ) @>{\hat{p}_!}>> \KZ l @>{\cup a_2}>> \KZ a_2 l
\end{CD}
\nonumber
\end{equation}
The map $\phi_1$ is injective because of the previous part of
the sequence. The last map is injective and surjective since
$a_2 \cdot l \neq 0$ because we chose 
$q_2 = q_2 \circ \phi$ above\footnote{See paragraph after
the proof of the previous Lemma.}.
Hence, by exactness, the map $\hat{p}_!$ is zero and 
$H^2(\hat{E}_{3,2},\KZ) \simeq \KZ \hat{p}^{\ast}(a_1).$
\item{\underline{Degree 3}} Consider the Gysin sequence for
$\hat{E}_{3,2}$ beginning at $H^1(R_{3,2},\KZ)$
\begin{equation}
\begin{CD}
\KZ l @>{\phi_1=\cup a_2}>> \KZ a_2 l @>{\hat{p}^{\ast}}>> 
H^3(E_{3,2},\KZ) @>{\hat{p}_!}>> \KZ a_1 \oplus \KZ a_2 
@>{\cup a_2}>> \KZ a_1^2 \oplus \KZ a_2^2 \oplus \KZ x
\end{CD}
\nonumber
\end{equation}
The map $\phi_1$ is an isomorphism and hence $\hat{p}^{\ast} = 0.$
By exactness, $H^3(\hat{E}_{3,2},\KZ) \simeq \KZ \hat{h}$ and
$\hat{p}_!(\hat{h}) = a_1.$
\end{itemize}
\end{proof}

\begin{theorem}
\begin{enumerate}
\item The bundle $\hat{p}:\hat{E}_{3,2} \to R_{3,2}$ described above and
cohomology classes $\hat{b} \in H^2(\hat{E}_{3,2},\KZ)$ and 
$\hat{h} \in H^2(\hat{E}_{3,2},\KZ)$ are such that 
$T_{3,2}$ classifies the triple $([\hat{p}],\hat{b},\hat{h})$
over $R_{3,2}.$
\item Let $\tilde{T}: \hat{E}_{3,2} \to E_{3,2}$ be the map
covering $T_{3,2}:R_{3,2} \to R_{3,2},$ then $\tilde{T}$ 
acts on the generators of the low-dimensional cohomology of
$H^{\ast}(E_{3,2},\KZ)$ as follows
\begin{itemize}
\item $\tilde{T}^{\ast}(y) = 0$
\item $\tilde{T}^{\ast}(p^{\ast}(a_2)) = \hat{p}^{\ast}(a_1)$
\item $\tilde{T}^{\ast}(p^{\ast}(a_2 l)) = 0$
\item $\tilde{T}^{\ast}(h) = \hat{h}$
\item The T-dual of $b$ is of the form $k p^{\ast}(a_1)$ for some
$k.$ Hence, the T-dual of a triple $(p,b,H)$ is a triple of
the form $(\hat{p},p^{\ast}(y),\hat{H}).$ 
\end{itemize}
\item $T_{3,2}^2$ cannot be homotopic to the identity,
i.e. T-duality for triples is not involutive.
\end{enumerate}
\end{theorem}
\begin{proof}
\begin{enumerate}
\item We consider the triple $([\hat{p}], \hat{b}, \hat{h})$ over
$R_{3,2}.$ This triple is classified by a map $f:R_{3,2} \to R_{3,2}.$
Then $f^{\ast}(a_1) = a_2.$ 
It is clear that $f^{\ast}E_{3,2} = \hat{E}_{3,2}.$
In addition if we forget the $B$-class, the action of $f$ on the
pair $([p],h)$ is to transform it into $([\hat{p}],\hat{h}).$ 
This implies that $q_3 \circ f = T_3 \circ q_3.$
Hence $f$ is the homotopic to $T_{3,2}$ by the uniqueness
of the pullback square described in Thm.\ (\ref{PBR32}).

\item 


Now, in the above two sequences, $T_{3,2}^{\ast}$ induces a natural map 
$(a_1,a_2) \to (a_2,a_1)$ on $H^2(R_{3,2},\KZ)$, further the 
induced maps on the remaining cohomology groups except $H^2(E_{3,2},\KZ)$ 
are the identity.

We have that $\hat{p}^{\ast} \circ T^{\ast}_{3,2} = 
{\tilde{T}}^{\ast} \circ p^{\ast}.$ As a result, 
$\tilde{T}(p^{\ast}(a_2)) = \hat{p}^{\ast}(a_1).$
Also, we have 
$T^{\ast}_{3,2} \circ p_! = \hat{p}_! \circ {\tilde{T}}^{\ast}.$
Hence 
$\hat{p}_! \circ {\tilde{T}}^{\ast}(b) = T^{\ast}_{3,2}(l).$ 
However, from the Gysin sequence for $\hat{E}_{3,2},$ we have 
$$
\KZ \hat{p}^{\ast}(a_1) \stackrel{\hat{p}_!}{\to} \KZ l 
\stackrel{\cup a_2}{\to} \KZ a_2 l \to \KZ \hat{h}.
$$
Now the map $H^1(R_{3,2},\KZ) \to H^3(R_{3,2},\KZ)$ is an 
isomorphism and hence $\hat{p}_!$ is zero.
As a result, $T^{\ast}_{3,2}(l)=0,$ and hence ${\tilde{T}}^{\ast}(y) = 0.$

Consider the above pair of Gysin sequences beginning at 
$H^3(R_{3,2},\KZ)$:
We have for $E_{3,2}$
$$
\KZ l \stackrel{\cup a_1}{\to} \KZ a_2 l \stackrel{p^{\ast}}{\to}
\KZ p^{\ast}(a_2 l) \oplus \KZ h \stackrel{p_!}{\to} \KZ a_1 \oplus \KZ a_2
\stackrel{\cup a_1}{\to} \KZ a_1^2 \oplus \KZ a_2^2 \oplus \KZ x
$$
and for $\hat{E}_{3,2}$
$$ 
\KZ l \stackrel{\cup a_2}{\to} \KZ a_2 l \stackrel{\hat{p}^{\ast}}{\to}
\KZ \hat{h} \stackrel{\hat{p}_!}{\to} \KZ a_1 \oplus \KZ a_2
\stackrel{\cup a_2}{\to} \KZ a_1^2 \oplus \KZ a_2^2 \oplus \KZ x
$$
From the above action of $T^{\ast}_{3,2}$ on $H^2(R_{3,2},\KZ)$
it is clear that $\tilde{T}^{\ast}(h) = \hat{h}.$ Also, since
$p^{\ast}$ is an isomorphism and $\hat{p}^{\ast}$ is zero, it is clear that
$\tilde{T}^{\ast}(p^{\ast}(a_2 l)) = 0.$
From the above, we have
\begin{itemize}
\item $\tilde{T}^{\ast}(y) = 0$
\item $\tilde{T}^{\ast}(p^{\ast}(a_2)) = \hat{p}^{\ast}(a_1)$
\item $\tilde{T}^{\ast}(p^{\ast}(a_2 l)) = 0$
\item $\tilde{T}^{\ast}(h) = \hat{h}.$
\item This is clear since, by the above, we know that 
$\hat{p}_!(\tilde{T}(b)) = 0.$ The result for the form of the T-dual
triple is obvious.
\end{itemize}
\item From the above, it is clear that 
\begin{itemize}
\item $T_{3,2}^{\ast}(l) = 0$
\item $T_{3,2}^{\ast}(a_1,0) = (0,a_2)$ and
$T_{3,2}^{\ast}(0,a_2) = (a_1,0).$
\item $T_{3,2}^{\ast}(a_2 l) = 0$
\end{itemize}
Therefore $T_{3,2}^2$ cannot be homotopic to the identity.
\end{enumerate}
\end{proof}

\begin{theorem}
\begin{enumerate}
\item There is a natural map $\gamma:R_3 \to R_{3,2}$ and hence
there is a natural triple associated to a pair $(p,H)$ over any
space, namely the triple $(p,p^{\ast}(p_!H),H).$
\item The T-dual of the triple in the previous item is the triple 
$(\hat{p},\hat{p}^{\ast}(\hat{p}_!\hat{H}),\hat{H}),$
that is, $T_{3,2} \circ \gamma = \gamma \circ T_3.$
\item Each class $\alpha$ in $H^2(K(\KZ^2,2),\KZ)$ gives a
map $\gamma_\alpha:R_3 \to R_{3,2}.$ Each of these maps is an 
inverse to the forgetful map $R_{3,2} \to R_3$ which `forgets' the
class $b.$
\end{enumerate}
\end{theorem}
\begin{proof}
\begin{enumerate}
\item Note that we have a natural map $R_3 \to K(\KZ,2) \times K(\KZ,2)$
which is given by sending $(p,H) \to (p,p_!(H)).$  We have a natural
identification (up to homotopy) of $K(\KZ,2) \times K(\KZ,2)$ with 
$\tilde{R}_2$ and, as a result, there is a natural map 
$\gamma:R_3 \to R_2.$ 
We may naturally associate to any $3$-pair $(p,H)$ over
a space $X$ a $2$-pair $(p,p^{\ast}p_!H)$ over that space. 
Note that the map $\gamma$ commutes with the natural maps 
$\mu:R_3 \to K(\KZ,2)$ and $\nu:R_2 \to K(\KZ,2).$ Hence,
there is a natural lift $\eta$ of $\gamma$ to $R_{3,2}$ given by
$\eta:R_3 \to R_{3,2}$ such that 
$\eta(a) = (a, \gamma(a)) \in R_3 \times R_2.$ 
This gives a natural triple associated to any pair $(p,H):$
To $(p,H)$ we associate $(p,p^{\ast}p_!H,H).$ 

\item We know that under T-duality, the bundle $E_3$ on $R_3$ is mapped
to the bundle $\hat{E}_3$ on $R_3.$ If we examine the cohomology of 
these bundles, we see that the T-dual of the triple 
$(p,p^{\ast}p_!h,h)$ should be a triple of the form
$(q, k.q^{\ast}q_!\hat{h}, \hat{h}), k \in \KZ.$  By the above
theorem, the T-dual is exactly the triple $(q, q^{\ast}q_!\hat{h}, \hat{h}).$
This implies that the map $\gamma$ commutes with $T_{3,2}$ and $T_3:$
That is, $T_{3,2} \circ \gamma = \gamma \circ T_3.$

\item Note that each map $f:K(\KZ^2,2) \to K(\KZ^2,2)$ up to 
homotopy gives a lift of the form above, if the map is such
that $f(a,b) = (a, \alpha(a,b)).$ Thus, each element $\alpha$ in 
$H^2(K(\KZ^2,2),\KZ)$ gives such a lift.
It is clear that each of these lifts correspond to triples of the form 
$(p, k.p^{\ast}(a_1), h), k \in \KZ$ on $R_3.$
Note that each of these lifts is an inverse to the forgetful map 
$R_{3,2} \to R_3$ which `forgets' the element $b$ of a triple. 
\end{enumerate}
\end{proof}

We may now prove the result we conjectured in Sec.\ (\ref{SecInvol}):
\begin{corollary}
Let $W$ be a CW-complex and $p:E \to W$ a principal $S^1$-bundle
over $W.$ Let $b \in H^2(E,\KZ)$ and $H \in H^3(E,\KZ)$ be the
$B$-class and $H$-flux on the bundle. Let $([q], b^{\#},H^{\#})$
be the T-dual triple.
Topological T-duality induces a bijection between the sets
$$
\{b + l p^{\ast}p_!(H)\} \protect{\mbox{ and }}
\{b^{\#} + m q^{\ast}{q_!(H^{\#})}\}.
$$
\end{corollary}
\begin{proof}
Consider the universal bundle $E_{3,2}$ over $R_{3,2}$ and the
T-dual bundle $\hat{E}_{3,2}$ over $R_{3,2}.$ 
The result is obvious for the universal bundle since 
$b^{\#}$ is always of the form $k \cdot q^{\ast}{q_!(H^{\#})}, k \in \KZ.$ 
By pullback, the result is true for $X.$
\end{proof}

We now return to the question raised in the paragraph after 
Lemma (\ref{LemATMP}) in Sec.\ (\ref{SecBr}). By Lemma (\ref{LemATMP}) 
and Sec.\ (\ref{SecModel}), for a given space $X,$ at fixed $H$
this is equivalent to knowing the image of the map $T_{3,2}$ studied above. The above
Corollary indicates that all automorphisms with Phillips-Raeburn invariant in 
a given coset map into automorphisms with Phillips-Raeburn invariant in 
another coset. This constrains the map and determines it in a large variety
of cases as was seen previously in Sec.\ (\ref{SecInvol}).

\section{Acknowledgements}
I thank Professor Jonathan Rosenberg, University of Maryland, College Park,
for encouraging me to work on this problem and for 
many useful discussions.  I acknowledge financial
aid from a Research Assistantship supported by the
NSF Grant DMS-0504212 during a portion of this work. 

I thank Professor Peter Bouwknegt, ANU, for useful discussions on the
physical meaning of the triples studied here.
I acknowledge financial aid
from a Postdoctoral Fellowship supported by 
the ARC Discovery Project `Generalized Geometries and 
their Applications' during a portion of this work.

I thank the Department of Mathematics, Harish-Chandra Research Institute,
Allahabad, for support from a postdoctoral fellowship during a 
portion of this work.

\end{document}